\documentclass[reqno]{amsart}
\pdfoutput=1
\usepackage[foot]{amsaddr}
\usepackage[UKenglish]{babel}
\usepackage[T1]{fontenc}
\usepackage{lmodern}
\usepackage{microtype}
\newif\ifams
\amstrue
\let\authoranonymous\relax
\usepackage[bb=boondox]{mathalpha}
\renewcommand{\paragraph}{\subsubsection*}
\newenvironment{claimproof}{\begin{proof}}{\end{proof}}
\newcommand{\claimqedhere}{\qedhere}
\usepackage[numbers,compress]{natbib}
\makeatletter
\def\NAT@spacechar{~}\makeatother
\usepackage[pdfusetitle]{hyperref}
\hypersetup{
    colorlinks,
    linkcolor={red!20!black},
    citecolor={green!20!black},
    urlcolor={blue!20!black}
}
\usepackage[small]{complexity}
\usepackage{booktabs}
\usepackage{amssymb,mathtools,amsmath,amsthm,mathrsfs,thmtools,stmaryrd,xspace}
\DeclareFontFamily{U}{mathb}{}
\DeclareFontShape{U}{mathb}{m}{n}{<-5.5> mathb5 <5.5-6.5> mathb6 
<6.5-7.5> mathb7 <7.5-8.5> mathb8 <8.5-9.5> mathb9 <9.5-11> mathb10 
<11-> mathb12}{}
\DeclareSymbolFont{mathb}{U}{mathb}{m}{n}
\DeclareFontSubstitution{U}{mathb}{m}{n}
\DeclareMathSymbol{\sqsubsetneq}{\mathbin}{mathb}{"88} 
\usepackage{tikz}
\usetikzlibrary{positioning}
\usetikzlibrary{shapes.geometric}
\usetikzlibrary{automata}
\usetikzlibrary{arrows}
\tikzstyle{square}=[rounded corners=4pt,regular polygon,regular polygon
  sides=4,minimum size=.6cm,draw=gray!90,inner sep=0pt,fill=gray!20,thick]
\tikzstyle{lsquare}=[rounded corners=4pt,draw,minimum
  width=.6cm,minimum height=.45cm,draw=gray!90,inner sep=0pt,fill=gray!20,thick]
\tikzstyle{every node}=[font=\footnotesize]
\tikzstyle{every edge}=[draw,>=stealth',shorten >=1pt,thin]

\providecommand{\urlstyle}[1]{}
\providecommand{\doi}[1]{\href{https://doi.org/#1}{\nolinkurl{doi:#1}}}
\usepackage{cleveref}

\newcommand{\eqby}[1]{\mathrel{\raisebox{-.1ex}{\ensuremath{\stackrel{\raisebox{-.25ex}{\scalebox{.5}{\upshape\textrm{#1}}}}{=}}}}}
\newcommand{\eqdef}{\eqby{def}}
\newcommand{\fdim}[1]{\mathrm{fdim}\,#1}
\newcommand{\fin}{\mathrm{fin}}
\newcommand{\idealto}{\leadsto}

 \usepackage{subcaption}
\newtheorem{theorem}{Theorem}[section]
\newtheorem{lemma}[theorem]{Lemma}
\newtheorem{fact}[theorem]{Fact}
\newtheorem{corollary}[theorem]{Corollary}
\newtheorem{proposition}[theorem]{Proposition}
\newtheorem{claim}[theorem]{Claim}

\theoremstyle{definition}
\newtheorem*{problem}{Problem}
\newtheorem{definition}[theorem]{Definition}
\newtheorem{example}[theorem]{Example}
\newtheorem{remark}[theorem]{Remark}
\usepackage{url,hyperref}
\crefrangelabelformat{equation}{(#3#1#4--#5#2#6)}
\crefname{section}{\defaultS\!}{\defaultS\!}
\Crefname{section}{Section}{Sections}
\crefname{subsection}{\defaultS\!}{\defaultS\!}
\Crefname{subsection}{Section}{Sections}
\crefname{subsubsection}{\defaultS\!}{\defaultS\!}
\Crefname{subsubsection}{Section}{Sections}
\crefname{theorem}{Theorem}{theorems}
\Crefname{theorem}{Theorem}{Theorems}
\crefname{lemma}{Lemma}{lemmata}
\Crefname{lemma}{Lemma}{Lemmata}
\crefname{fact}{Fact}{facts}
\Crefname{fact}{Fact}{Facts}
\crefname{corollary}{Corollary}{corollaries}
\Crefname{corollary}{Corollary}{Corollaries}
\crefname{proposition}{Proposition}{propositions}
\Crefname{proposition}{Proposition}{Propositions}
\crefname{claim}{Claim}{claims}
\Crefname{claim}{Claim}{Claims}
\crefname{definition}{Definition}{definitions}
\Crefname{definition}{Definition}{Definitions}
\crefname{example}{Example}{examples}
\Crefname{example}{Example}{Examples}
\crefname{remark}{Remark}{remarks}
\Crefname{remark}{Remark}{Remarks}
\crefname{figure}{Figure}{figures}
\Crefname{figure}{Figure}{Figures}
\crefname{table}{Table}{tables}
\Crefname{table}{Table}{Tables}
\crefname{equation}{equation}{equations}
\Crefname{equation}{Equation}{Equations}
\def\keywords{\smallskip\noindent\textsc{Keywords.}~~}
\def\vec#1{\mathchoice{\mbox{\boldmath$\displaystyle#1$}}
{\mbox{\boldmath$\textstyle#1$}}
{\mbox{\boldmath$\scriptstyle#1$}}
{\mbox{\boldmath$\scriptscriptstyle#1$}}}
\begin{document}
\newcommand{\mydim}{d}
\title[On the Length of Strongly Monotone Descending Chains over \texorpdfstring{$\+N^{\mydim}$}{ℕᵈ}]{On the Length of Strongly Monotone\\Descending Chains over {\boldmath $\mathbbb{\+N}^{\mydim}$}}
\author[S.~Schmitz and L.~Sch\"utze]{Sylvain Schmitz$^{1}$ and Lia Sch\"utze$^2$}
\address{$^1$~Universit\'e Paris Cité, CNRS, IRIF, Paris, France}
\address{$^2$~Max Planck Institute for Software Systems (MPI-SWS), Kaiserslautern, Germany}
\begin{abstract}
  A recent breakthrough by \citeauthor*{kunnemann23} (ICALP 2023) bounds
the running time for the coverability problem in $d$-dimensional
vector addition systems under unary encoding to $n^{2^{O(d)}}$,
improving on \citeauthor{rackoff78}'s $n^{2^{O(d\lg d)}}$ upper bound
(\textit{Theor.\ Comput.\ Sci.} 1978), and provides conditional
matching lower bounds.

In this paper, we revisit \citeauthor*{lazic21}' ``ideal view'' of the
backward coverability algorithm (\textit{Inform.\ Comput.} 2021) in
the light of this breakthrough.  We show that the controlled strongly
monotone descending chains of downwards-closed sets over $\+N^d$ that
arise from the dual backward coverability algorithm
of \citeauthor*{lazic21} on $d$-dimensional unary vector addition
systems also enjoy this tight $n^{2^{O(d)}}$ upper bound on their
length, and that this also translates into the same bound on the running
time of the backward coverability algorithm.

Furthermore, our analysis takes place in a more general setting than
that of \citeauthor{lazic21}, which allows to show the same results
and improve on the \ComplexityFont{2EXPSPACE} upper bound derived
by \citeauthor*{benedikt17} (LICS~2017) for the coverability problem
in invertible affine nets.
 
\end{abstract}
\maketitle

\section{Introduction}
\paragraph{Well-Quasi-Orders\nopunct} (wqo for short) are a
notion from order theory~\cite{kruskal72,SS2012} that has proven
very effective in many areas of mathematics, logic,
combinatorics, and computer science in order to establish finiteness
statements.  For instance, in the field of formal verification, they
provide the termination arguments for the generic algorithms
for \emph{well structured transition
systems}~\cite{abdulla00,finkel01}, notably the \emph{backward
coverability algorithm} for deciding safety
properties~\cite{arnold78,abdulla00,finkel01}.

In full generality, one cannot extract complexity bounds from
wqo-powered termination proofs.  Nevertheless, in an algorithmic
setting, one can ``instrument'' wqos by considering
so-called \emph{controlled sequences}~\cite{SS2012,hdr/Schmitz17}, and
new tight complexity upper bounds for wqo-based algorithms now appear
on a regular basis~\cite[e.g.,][for a few recent
examples]{icalp/Schmitz19,bala20,bala21,bala22,guillou23}.

Those complexity upper bounds are however astronomically high, and
sometimes actually way too high for the problem at hand.  An
emblematic illustration of this phenomenon is the backward
coverability algorithm for vector addition systems (VAS), which was
shown to run in double exponential time by \citet{bozzelli11} based on
an original analysis due to \citet{rackoff78}: the corresponding
bounds over the wqo~$\+N^d$ are Ackermannian~\cite{figueira11}.

\paragraph{Descending Chains.}  One way pioneered by \citet{lazic21}
to close such complexity gaps while retaining some of the wide
applicability of wqos and well structured transition systems is to
focus on the descending chains of downwards closed sets over the wqo
at hand.  Indeed, one of the equivalent characterisations of wqos is
the \emph{descending chain condition}~\cite{kruskal72,SS2012}, which
guarantees that those descending chains are finite.

In themselves, descending chains are no silver bullet: e.g., the
controlled descending chains over $\+N^d$ are also of Ackermannian
length~\cite[Thm.~3.10]{lazic21}.  Nevertheless, these chains
sometimes exhibit a form of ``monotonicity,'' which yields vastly
improved upper bounds.  When applied to a \emph{dual} version of the
backward coverability algorithm in well structured transition systems,
this allows to recover the same double exponential time upper bound as
in~\citep{bozzelli11,rackoff78} for the VAS coverability problem,
along with tight upper bounds for coverability in several VAS
extensions.  The same framework was also the key to establishing tight
bounds for coverability in $\nu$-Petri nets~\cite{LazicS16}.  As a
further testimony to the versatility of the approach,
\citeauthor*{benedikt17} use it in~\citep{benedikt17} to derive
original upper bounds for problems on invertible polynomial automata
and invertible affine nets, in a setting that is not strictly speaking
one of well structured transition systems.

\paragraph{Fine-grained Bounds for VAS Coverability.} The coverability
problem in VAS is well-known to be \EXPSPACE-complete, thanks
to \citeauthor{rackoff78}'s \citeyear{rackoff78} upper bound matching
a \citeyear{lipton76} lower bound by \citet{lipton76}.  The main
parameter driving this complexity is the dimension of the system: the
problem is in pseudo-polynomial time in fixed dimension~$d$; more
precisely, \citeauthor{rackoff78}'s analysis yields a $n^{2^{O(d\lg
d)}}$ deterministic time upper bound for $d$-dimensional VAS encoded
in unary~\citep{rosier86}, by proving the same bound on the length of
a covering execution of minimal length.  Here, there is a discrepancy
with the $n^{2^{\Omega(d)}}$ lower bound on the length of that
execution in \citeauthor{lipton76}'s construction---a discrepancy that
was already highlighted as an open problem in the early 1980's
by \citet{mayr82}, and settled in the specific case of reversible
systems by \citet{koppenhagen00}.  The upper bounds of
both \citet{bozzelli11} and \citet{lazic21} on the complexity of the
backward coverability algorithm inherit from \citeauthor{rackoff78}'s
$n^{2^{O(d\lg d)}}$ bound and suffer from the same discrepancy.

This was the situation until \citeauthor*{kunnemann23} showed an
$n^{2^{O(d)}}$ upper bound on the length of minimal covering
executions of unary encoded $d$-dimensional VAS,
matching \citeauthor{lipton76}'s lower
bound~\citep[Thm.~3.3]{kunnemann23}.  This directly translates into a
deterministic algorithm with the same upper bound on the running
time~\citep[Cor.~3.4]{kunnemann23}.  Furthermore, assuming the
exponential time hypothesis, \citeauthor{kunnemann23} also show that
there does not exist a deterministic $n^{o(2^d)}$ time algorithm
deciding coverability in unary encoded $d$-dimensional
VAS~\citep[Thm.~4.2]{kunnemann23}.

\paragraph{Thinness.} The improved upper bound relies on the notion
of a \emph{thin} vector in~$\+N^d$~\citep[Def.~3.6]{kunnemann23}
(somewhat reminiscent of the ``extractors'' of \citet{leroux13}).  The
proof of \citep[Thm.~3.3]{kunnemann23} works by induction on the
dimension~$d$. By splitting a covering execution of minimal length at
the first non-thin configuration, \citeauthor{kunnemann23} obtain a
prefix made of distinct thin configurations (which must then be of
bounded length), and a suffix starting from a configuration with some
components high enough to be disregarded, hence that can be treated as
an execution in a VAS of lower dimension, on which the induction
hypothesis applies.

\paragraph{Contributions.}  In this paper, we show that the improved
$n^{2^{O(d)}}$ upper bound of \citet{kunnemann23} also applies to the
number of iterations of the backward coverability algorithm for
$d$-dimensional VAS encoded in unary (see \cref{th-vas}).  In order to
do so, one could reuse the approach of \citet{bozzelli11} to lift the
improved bound from the length of minimal covering executions to the
running time of the backward coverability algorithm, but here we aim
for the generality of the framework of~\cite{lazic21}.

Our main contribution is thus to show in \cref{sec-thin} that the
upper bounds on the length of strongly monotone controlled descending
chains of downwards closed sets over~$\+N^d$---which include those
constructed during the running of the backward coverability algorithm
for VAS---can be improved similarly (see \cref{th-thin}) when focusing
on a suitably generalised notion of thinness.  As a byproduct, we
observe that thinness is an inherent property of such chains
(see \cref{cor-thin}), rather than an \emph{a priori} condition
that---almost magically---yields the improved bound.

We apply our results to the coverability problem in vector addition
systems in \cref{sec-vas}---thus providing as promised an alternative
to applying \citeauthor{bozzelli11}'s approach
to \citeauthor{kunnemann23}'s results---and show that the backward
coverability algorithm runs in time~$n^{2^{O(d)}}$
(see \cref{cor-bc-vas}) and is therefore conditionally optimal
by~\citep[Thm.~4.2]{kunnemann23}.

As a further demonstration of the versatility of our results, we show
in \cref{sec-an} how to apply them to invertible affine nets, a
generalisation of vector addition systems introduced
by \citet{benedikt17}, and a good showcase for our techniques.  We
obtain the same bounds for their coverability problem as in the case
of vector addition systems (see \cref{th-an,cor-an-bc}), and thereby
improve on the \ComplexityFont{2EXPSPACE} upper bound
of~\citep{benedikt17} by showing that the problem is
actually \ComplexityFont{EXPSPACE}-complete (see \cref{cor-an}).
Along the way, we will see that the improved upper bounds also
apply for other VAS extensions, for which Rackoff's proof scheme had
been successfully adapted (see \cref{rk-abvas,rk-sian}), namely
strictly increasing affine nets~\citep{bonnet12}, branching
VAS~\citep{demri12}, and alternating VAS~\citep{courtois14}.
 
\section{Well-Quasi-Orders and Ideals}\label{sec-ideals}
We start by introducing the necessary background on
well-quasi-orders, descending chains, and order ideals.

\paragraph{Well-Quasi-Orders.}  A \emph{quasi-order} $(X,{\leq})$
comprises a set~$X$ and a transitive reflexive relation
${\leq}\subseteq X\times X$.  For a subset $S\subseteq X$,
its \emph{downward closure} is the set of elements smaller or equal to
some element in~$S$, i.e., ${\downarrow}S\eqdef\{x\in X\mid\exists
y\in S\mathbin.x\leq y\}$.  When $S=\{y\}$ is a singleton, we note
${\downarrow}y$ for this set.  A subset $S\subseteq X$
is \emph{downwards-closed} if $S={\downarrow}S$.
A \emph{well-quasi-order} is a quasi-order $(X,{\leq})$ such that all
the \emph{descending chains}
\begin{equation}\label{eq-desc-seq}
  D_0 \supsetneq D_1 \supsetneq D_2 \supsetneq\cdots
\end{equation}
of downwards-closed subsets $D_k\subseteq X$ are
finite~\citep{kruskal72,SS2012}.

Conversely, the \emph{upward closure} of a subset $S\subseteq X$ is
${\uparrow}S\eqdef\{x\in X\mid\exists y\in S\mathbin.y\leq x\}$, and
$S$ is \emph{upwards-closed} if $S={\uparrow}S$.  The complement
$X\setminus D$ of a downwards-closed set~$D$ is upwards-closed (and
conversely), hence wqos have the \emph{ascending chain condition} for
chains $U_0\subsetneq U_1\subsetneq\cdots$ of upwards-closed sets:
they are necessarily finite.
Furthermore, any upwards-closed set~$U$ over a wqo has a \emph{finite
basis}~$B$ such that $U={\uparrow}B$~\citep{kruskal72,SS2012}; without
loss of generality, we can take the elements of~$B$ to be minimal and
mutually incomparable in~$U$.

A well-studied wqo is $(\+N^d,{\sqsubseteq})$ the set of
$d$-dimensional vectors of natural numbers along with the
component-wise (aka product) ordering~\citep{dickson13};
see \cref{fig-vas} for an illustration of a descending chain over
$\+N^2$, which happens to be produced by the backward coverability
algorithm for a vector addition system~\citep[Ex.~3.6]{lazic21}.

\begin{figure}[tbp]
  \captionsetup[subfigure]{labelformat=empty,justification=justified,singlelinecheck=false}
  \centering
  \newcommand\onestep{.34}
  \newcommand\oneinf{9*\onestep+.2}
  \newcommand{\point}[3]{\draw[fill=#1!40,draw=#1!60!black] (#2*\onestep,#3*\onestep) circle (2pt)}
  \newcommand\onegrid{\draw[step=\onestep,gray!40,very thin] (0,0) grid (\oneinf,5*\onestep+.2);
    \draw (0,0) -- (0,5*\onestep+.2);
    \draw (0,0) -- (\oneinf,0);
    \foreach \i in {1,...,5}
      \draw[dotted,gray!60] (\oneinf,\i*\onestep) -- (\oneinf+\onestep,\i*\onestep);
    \draw[dotted,black] (\oneinf,0) -- (\oneinf+\onestep,0);
}
  \begin{subfigure}[t]{.33\textwidth}\centering
    \begin{tikzpicture}
\onegrid
    \foreach \i in {0,...,10}
      \foreach \j in {0,...,4}
        \point{violet}{\i}{\j};
    \end{tikzpicture}
    \caption{\footnotesize$D_0=\{(\omega,4)\}$}
  \end{subfigure}~
  \begin{subfigure}[t]{.33\textwidth}\centering
    \begin{tikzpicture}
\onegrid
    \foreach \i in {0,1}
      \point{violet!55!blue}{\i}{4};
    \foreach \i in {0,...,10}
      \foreach \j in {0,...,3}
        \point{violet!55!blue}{\i}{\j};    
    \end{tikzpicture}
    \caption{\footnotesize$D_1=\{(1,4),(\omega,3)\}$}
  \end{subfigure}~
  \begin{subfigure}[t]{.33\textwidth}\centering
    \begin{tikzpicture}
\onegrid
    \foreach \i in {0,1}
      \point{blue!90!cyan}{\i}{4};
    \foreach \i in {0,...,3}
      \point{blue!90!cyan}{\i}{3};
    \foreach \i in {0,...,10}
      \foreach \j in {0,...,2}
        \point{blue!90!cyan}{\i}{\j};   
    \end{tikzpicture}
    \caption{\footnotesize$D_2=\{(1,4),(3,3),(\omega,2)\}$}
  \end{subfigure}\\
  \begin{subfigure}[t]{.33\textwidth}\centering
    \begin{tikzpicture}
\onegrid
    \foreach \i in {0,1}
      \point{blue!40!cyan}{\i}{4};
    \foreach \i in {0,...,3}
      \point{blue!40!cyan}{\i}{3};
    \foreach \i in {0,...,5}
      \point{blue!40!cyan}{\i}{2};
    \foreach \i in {0,...,10}
      \foreach \j in {0,...,1}
        \point{blue!40!cyan}{\i}{\j}; 
    \end{tikzpicture}
    \caption{\footnotesize$D_3=\{(1,4),(3,3),(5,2),(\omega,1)\}$}
  \end{subfigure}~
  \begin{subfigure}[t]{.33\textwidth}\centering
    \begin{tikzpicture}
\onegrid
    \foreach \i in {0,1}
      \point{blue!20!cyan}{\i}{4};
    \foreach \i in {0,...,3}
      \point{blue!20!cyan}{\i}{3};
    \foreach \i in {0,...,5}
      \point{blue!20!cyan}{\i}{2};
    \foreach \i in {0,...,7}
      \point{blue!20!cyan}{\i}{1};
    \foreach \i in {0,...,10}
      \point{blue!20!cyan}{\i}{0}; 
    \end{tikzpicture}
    \caption{\footnotesize$D_4=\{(1,4),(3,3),(5,2),\\\phantom{D_4=\{}(7,1),(\omega,0)\}$}
  \end{subfigure}~
  \begin{subfigure}[t]{.33\textwidth}\centering
    \begin{tikzpicture}
\onegrid
    \foreach \i in {0,1}
      \point{cyan!30!white}{\i}{4};
    \foreach \i in {0,...,3}
      \point{cyan!30!white}{\i}{3};
    \foreach \i in {0,...,5}
      \point{cyan!30!white}{\i}{2};
    \foreach \i in {0,...,7}
      \point{cyan!30!white}{\i}{1};
    \foreach \i in {0,...,9}
      \point{cyan!30!white}{\i}{0}; 
    \end{tikzpicture}
    \caption{\footnotesize$D_5=\{(1,4),(3,3),(5,2),\\\phantom{D_5=\{}(7,1),(9,0)\}$}
  \end{subfigure}

\caption{\label{fig-vas}A descending chain $D_0\supsetneq
  D_1\supsetneq\cdots\supsetneq D_5$ over~$\+N^2$~\citep[Ex.~3.6]{lazic21}.}
\end{figure}

\paragraph{Order Ideals.}  An \emph{order ideal} of $X$ is a downwards-closed
subset $I\subseteq X$, which is \emph{directed}: it is non-empty, and
if $x,x'$ are two elements of $I$, then there exists $y$ in~$I$ with
$x\leq y$ and $x'\leq y$.  Alternatively, order ideals are characterised as
the \emph{irreducible} non-empty downwards-closed sets of~$X$: an
order ideal is a non-empty downwards-closed set~$I$ with the property that,
if $I\subseteq D_1\cup D_2$ for two downwards-closed sets~$D_1$
and~$D_2$, then $I\subseteq D_1$ or $I\subseteq D_2$.

Over a wqo $(X,{\leq})$, any downwards-closed set $D\subseteq X$ has a
canonical decomposition as a finite union of order ideals
$D=I_1\cup\cdots\cup I_n$, where the $I_j$'s are mutually incomparable
for inclusion~\citep{bonnet75,goubault20}.  We write $I\in D$ if $I$
is an order ideal appearing in the canonical decomposition of~$D$, i.e., if
it is a maximal order ideal included in~$D$.  Then $D\subseteq D'$ if and only if,
for all~$I\in D$, there exists $I'\in D'$ such that $I\subseteq I'$.

\paragraph{Effective Representations over \ifams$\+N^d$\else{\boldmath$\mathbbb{\mathbb N}^d$}\fi.}  Over the wqo
$(\+N^d,{\sqsubseteq})$, the order ideals are exactly the sets of the form
${\downarrow}\vec{v}\cap\+N^d$ where $\vec v$ ranges over
$\+N_\omega^d\eqdef(\+N\uplus\{\omega\})^d$, where $\omega$ is a new
top element~\cite{goubault20}.
From here on, we will abuse notations and identify an order ideal $I$
of $\+N^d$ with the vector $\vec{v}$ in~$\+N^d_\omega$ such that
$I={\downarrow}\vec{v}\cap\+N^d$.    See for instance the 
decompositions in \cref{fig-vas}.

Let us introduce some notations for
the sets of \emph{infinite} and \emph{finite components} of~$I$,
namely
\begin{align}
  \omega(I)&\eqdef\{1\leq i\leq
  d\mid I(i)=\omega\}\;,&
  \fin(I)&\eqdef
  \{1\leq i\leq
  d\mid I(i)<\omega\}\;,\\
  \intertext{along with its \emph{dimension} and \emph{finite
  dimension}, respectively defined as}
  \dim I&\eqdef |\omega(I)|\;,&
  \fdim I&\eqdef |\fin(I)|\;.
\end{align}
Note that $\fin(I)=\{1,\dots,d\}\setminus\omega(I)$ and $\fdim
I=d-\dim I$.  For instance, the order ideal $I=(\omega,4)$ in the
decomposition of~$D_0$ in \cref{fig-vas} satisfies $\omega(I)=\{1\}$
and $\dim I=1$.

\medskip
The order ideals of~$\+N^d$, when represented as vectors
in~$\+N^d_\omega$, are rather easy to 
manipulate~\citep{goubault20} --- and thus so are the downwards-closed
subsets of~$\+N^d$
when represented as finite sets of vectors
in~$\+N^d_\omega$.  For instance,
\begin{itemize}
  \item $I\subseteq I'$ (as subsets of~$\+N^d$) if and only if $I\sqsubseteq I'$ (as vectors
     in~$\+N^d_\omega$) --- which incidentally entails $\omega(I)\subseteq\omega(I')$ and therefore
     $\dim I\leq \dim I'$; also note that, if $I\subseteq I'$ and $\dim I=\dim I'$,
    then $\omega(I)=\omega(I')$;
  \item the intersection of two order ideals is again an order ideal,
    represented by the vector $I\wedge I'$ defined by
    $(I\wedge I')(i)\eqdef\min(I(i),I'(i))$ for all $1\leq i\leq d$;
  \item the complement of an order ideal~$I$ is the upwards-closed set
    $\bigcup_{i\in\fin(I)}{\uparrow}\big((I(i)+1)\cdot\vec e_i\big)$,
    where $\vec e_i$ denotes the unit vector with ``$1$'' in
    coordinate~$i$ and ``$0$'' everywhere else.
\end{itemize}

\paragraph{Proper Ideals and Monotonicity.}  If $D\supsetneq D'$, then
there must be an order ideal $I\in D$ such that $I\not\in D'$.  Coming back
to a descending chain $D_0\supsetneq D_1\supsetneq\cdots\supsetneq
D_\ell$, we then say that an order ideal $I$ is \emph{proper} at
step~$k$, for $0\leq
k<\ell$, if $I\in D_k$ but $I\not\in D_{k+1}$; at each step~$0\leq
k<\ell$, there must be at least one proper order ideal.  In \cref{fig-vas},
$(\omega,4)$ is proper at step~0, and more
generally $(\omega,4-k)$ is the only proper order ideal
at step~$0\leq k<5$.

It turns out that the descending chains arising from some algorithmic
procedures, including the backward coverability algorithm for VAS,
enjoy additional relationships between their proper order ideals.
Over $(\+N^d,{\sqsubseteq})$, we say
that a descending chain $D_0\supsetneq D_1\supsetneq\cdots$ is
\begin{itemize}
\item \emph{strongly monotone}~\citep{novikov99,benedikt17} if, whenever an
  ideal $I_{k+1}$ is proper at some step~$k+1$,
  then there exists $I_{k}$ proper at step~$k$ such that $\dim
  I_{k+1}\leq\dim I_k$, and
\item in particular \emph{$\omega$-monotone}~\citep{lazic21} if, whenever an
  ideal $I_{k+1}$ is proper at some step~$k+1$,
  then there exists $I_{k}$ proper at step~$k$ such that
  $\omega(I_{k+1})\subseteq \omega(I_k)$.
\end{itemize}
The descending chain depicted in \cref{fig-vas} is
$\omega$-monotone---and thus strongly monotone\mbox{---} with
$\omega((\omega,4-(k+1)))\subseteq\omega((\omega,4-k))$
for all $4>k\geq 0$.

\paragraph{Controlled Sequences.}  While guaranteed to be
finite, descending chains
over a wqo can have arbitrary length.  Nevertheless, their length can
be bounded under additional assumptions.  We define the
\emph{size} of a downwards-closed subset of~$\+N^d$ and of an order ideal
of~$\+N^d$ as
\begin{align}
  \|D\|&\eqdef \max_{I\in D}\|I\|\;,  & \|I\|&\eqdef\max_{i\in\fin(I)}I(i)\;.
\end{align}
In \cref{fig-vas}, $\|D_0\|=\|D_1\|=\|D_2\|=4$, $\|D_3\|=5$, $\|D_4\|=7$,
and $\|D_5\|=9$.

Given a \emph{control} function $g{:}\,\+N\to\+N$, which will always be monotone
(i.e., $\forall x\leq y.g(x)\leq g(y)$) and expansive (i.e., $\forall x.x\leq
g(x)$) along with an \emph{initial size} $n_0\in\+N$, we say that a
descending chain $D_0\supsetneq D_1\supsetneq\cdots$ over $\+N^d$ is
\ifams(asymptotically) \fi \emph{$(g,n_0)$-controlled} if, for all $k\geq 0$,
\begin{equation}
 \|D_k\|\leq g^k(n_0)
\end{equation}
where $g^k(n_0)$ is the $k$th iterate of $g$ applied to
$n_0$~\citep{hdr/Schmitz17}.  In particular, $\|D_0\|\leq n_0$
initially.  In \cref{fig-vas}, the descending chain is
$(g,4)$-controlled for $g(x)\eqdef x+1$.
 
\section{Main Result}\label{sec-thin}
In this section, we establish a new bound on the length of controlled
strongly monotone descending sequences.  This relies on a 
generalisation of the notion of \emph{thinness} from
\citet[Def.~3.6]{kunnemann23} (see \cref{sub-thin-def}), before we can
apply thinness in the setting of strongly monotone descending
chains and prove our main result in \cref{sub-thin-lem}.

\subsection{Thinness}\label{sub-thin-def}
Fix a control function~$g$, an initial size~$n_0$, and a
dimension~$d\geq 0$.  Define inductively the bounds on sizes
$(N_{i})_{0\leq i\leq d}$ and lengths $(L_{i})_{0\leq i\leq d}$ as
follows\begin{align}
  N_{0} &\eqdef n_0\;,&
  N_{i+1} &\eqdef g^{L_{i}+1}(n_0)\;,\label{eqdef-Ni}\\
  L_{0} &\eqdef 0\;,&
  L_{i+1} &\eqdef L_{i}+\prod_{1\leq j\leq i+1}\hspace*{-.9em}(d-j+1)(N_{j}+1)\;.\label{eqdef-Li}
\end{align}
Beware the abuse of notation, as the bounds above depend on $(g,n_0)$
and~$d$, but those will always be clear from the context.

\begin{remark}[Monotonicity of $(N_i)_{0\leq i\leq d}$ and $(L_{i})_{0\leq i\leq d}$]\label{NL-mono}
By definition, for all $0\leq i<j\leq d$, $0\leq L_{i} < L_{j}$,
and because $g$ is assumed monotone expansive, $n_0\leq N_{i}\leq
N_{j}$.~\hfill\qedsymbol
\end{remark}

The following definition generalises \citep[Def.~3.6]{kunnemann23} to
handle order ideals and an arbitrary control function and initial size.
\begin{definition}[Thin order ideal]\label{def-thin} Let $(g,n_0)$ be a control function
  and initial size and $d>0$ a dimension.  An order ideal $I$ of~$\+N^d$ is
  \emph{thin} if there exists a bijection~$\sigma\colon\fin(I)
  \to\{1,\dots,\fdim I\}$ such that, for all $i\in\fin(I)$, $I(i)\leq
  N_{\sigma(i)}$.
\end{definition}
Observe that that, if $I'$ is thin, $I\subseteq I'$, and $\dim I=\dim
I'$, then $I$ is thin.
\begin{remark}[Number of thin order ideals]\label{eq-card-thin}
  There cannot be more than $\binom{d}{i}\cdot i!\cdot\prod_{1\leq
    j\leq i} (N_{j}+1)=\prod_{1\leq j\leq i}(d-j+1)(N_{j}+1)$
  distinct thin order ideals of finite dimension~$i$.  As will become
  apparent in the proofs, this is what motivates the
  definition in~\eqref{eqdef-Li}.
  \ifams\par
  Furthermore, if we let
    $\mathrm{Idl}^{\mathsf{thin}}(\+N^d)$ denote the set of thin order
    ideals of $\+N^d$, there is only one
  thin order ideal of finite dimension~$0$---namely
    $(\omega,\dots,\omega)$---, and
    \begin{align*} |\mathrm{Idl}^{\mathsf{thin}}(\+N^d)|&\leq 1+ \sum_{1\leq
  i\leq d}\prod_{1\leq j\leq i}(d-j+1)(N_j+1)\\ &=1+\sum_{1\leq i\leq
  d}(L_i-L_{i-1})\\
  &=1+L_d-L_0=1+L_d\;.\\[-3em]\end{align*}\bigskip\fi\hfill\qedsymbol
\end{remark}

\subsection{Thinness Lemma}\label{sub-thin-lem}

The crux of our result is the following lemma. 

\begin{lemma}[Thinness]\label{lem-thin}
  Consider a $(g,n_0)$-controlled strongly monotone descending
  chain~$D_0 \supsetneq D_1\supsetneq\cdots$ of downwards-closed
  subsets of~$\+N^d$. If $I_\ell$ is a proper order ideal at some
  step~$\ell$, then $I_\ell$ is thin and $\ell\leq
  L_{\fdim I_\ell}$.
\end{lemma}

The proof of \cref{lem-thin} proceeds by induction over the finite
dimension $\fdim I_\ell=d-\dim I_\ell$.  For the base
case where $I_\ell$ has full dimension $\dim I_\ell=d$, then $I_\ell=(\omega,\dots,\omega)$
is thin and $D_\ell=\+N^d$ is the full space, which can only occur at
step $\ell=0=L_{0}$.  For the induction step, we first establish
thinness with the following claim; note that, as just argued, an order
ideal of dimension~$d$ is necessarily thin.  We then follow with the
bound on $\ell$ to complete the \hyperref[prf:lem-thin]{proof
of \cref{lem-thin}}.

\begin{claim}\label{cl-thin} Let $0\leq d'<d$ and assume
  that \cref{lem-thin} holds for all proper order ideals~$I'$ of
  dimension~$\dim I'>d'$.  If $I$ is any (not necessarily proper) order ideal of dimension~$\dim
  I=d'$ appearing as a maximal ideal in the descending chain $D_0\supsetneq
  D_1\supsetneq\cdots$, then $I$ is thin.
\end{claim}
\begin{claimproof}[Proof of \Cref{cl-thin}]
    \ifams\renewcommand{\qedsymbol}{\tiny[\ref{cl-thin}]}\fi
  Let $k$ be a step where $I$ appears in the descending chain
  $D_0\supsetneq D_1\supsetneq\cdots$, i.e., $I\in D_k$, and let us
  write $I_k\eqdef I$.
If $k>0$, since $D_k\subseteq D_{k-1}$, there exists an order ideal
  $I_{k-1}\in D_{k-1}$ such that $I_k\subseteq I_{k-1}$.  If $k=0$, or
  by repeating this argument if $k>0$, we obtain a chain of order
  ideals (with decreasing indices)
  \begin{equation}\label{eq-order ideal-chain}
    I_k \subseteq I_{k-1}\subseteq\cdots\subseteq I_0
  \end{equation}
  where $I_m\in D_m$ for all $k\geq m\geq 0$.  Every order ideal in
  that chain must have dimension at least $\dim I_k=d'$ since they all
  contain $I_k$.  Two cases arise.
  \begin{enumerate}
  \item If every order ideal in
    the chain~\eqref{eq-order ideal-chain} has dimension $\dim I_k$,
    then because the descending chain $D_0\supsetneq
    D_1\supsetneq\cdots$ is $(g,n_0)$-controlled, we have $\|I_0\|\leq
    n_0=N_{0}$ and we know by \cref{NL-mono} that~$I_0$ is thin.  Since
    $I_k\subseteq I_0$ and $\dim I_k=\dim I_0$, $I_k$ is also
    thin.
  \item Otherwise there exists a first index $K$ along the
    chain~\eqref{eq-order ideal-chain} where the dimension increases,
    i.e., such that $\dim I_k<\dim I_K$ and $\dim I_m=\dim I_k$ for all
    $k\geq m>K$.  Then $I_K$ is proper, as otherwise $D_{K+1}$ would
    contain two distinct but comparable order ideals in its canonical
    decomposition, namely~$I_{K}$ and~$I_{K+1}$: indeed,
    $I_{K+1}\subseteq I_{K}$ and $\dim I_{K+1}=\dim I_k<\dim I_{K}$
    imply $I_{K+1}\subsetneq I_K$.  By 
    assumption, \cref{lem-thin} can be applied to $I_K$ of dimension
    $\dim I_K>\dim I_k=d'$, thus $I_K$ is thin and $K\leq L_{\fdim I_K}$.

    \medskip
    Let us now show that $I_{K+1}$ is thin, which will also yield that
    $I_k$ is thin since $I_k \subseteq I_{K+1}$ and $\dim I_k=\dim
    I_{K+1}$.

    Since $\dim I_{K+1}<\dim I_{K}$, we let $f\eqdef\dim I_K-\dim
    I_{K+1}=\fdim I_{K+1}-\fdim I_K>0$. As furthermore
    $I_{K+1}\subseteq I_K$, $\omega(I_{K+1})\subsetneq\omega(I_K)$ and
    we let
    $\{i_1,\dots,i_f\}\eqdef \omega(I_K)\setminus\omega(I_{K+1})=
    \fin(I_{K+1})\setminus\fin(I_K)$.
    
    Since $I_K$ is thin, there exists a bijection
    $\sigma\colon\fin(I_K)\to \{1,\dots,\fdim(I_K)\}$ such that
    $I_K(i)\leq N_{\sigma(i)}$ for all $i\in\fin(I_K)$.  We extend
    $\sigma$ to a bijection
    $\sigma'\colon \fin(I_K)\uplus\{i_1,\dots,i_f\}\to\{1,\dots,\fdim
    I_K+f\}$: we let $\sigma'(i)\eqdef\sigma(i)$ for all
    $i\in\fin(I_K)$, and $\sigma'(i_j)\eqdef\fdim I_K+j$ for all
    $1\leq j\leq f$.  Let us check that~$\sigma'$ witnesses the
    thinness of~$I_{K+1}$.
    \begin{itemize}
    \item Because $I_{K+1}\subseteq I_K$, for all those
      $i\in\fin(I_K)$, $I_{K+1}(i)\leq I_K(i)\leq
      N_{\sigma(i)}=N_{\sigma'(i)}$.
    \item Since $K+1\leq L_{\fdim I_K}+1$ and since the descending
      chain $D_0\supsetneq D_1\supsetneq\cdots$ is
      $(g,n_0)$-controlled, we have a bound of $g^{L_{\fdim
      I_K}+1}(n_0)=N_{\fdim I_K+1}$ on all the finite components of
      $I_{K+1}$, and in particular $I_{K+1}(i_j)\leq N_{\fdim I_K+1}$
      for all $1\leq j\leq f$. By \cref{NL-mono}, we conclude that
      $I_{K+1}(i_j)\leq N_{\fdim I_K+j}=N_{\sigma'(i_j)}$ for all
      $1\leq j\leq f$.\claimqedhere
    \end{itemize}
  \end{enumerate}
\end{claimproof}

\begin{proof}[Proof of \Cref{lem-thin}]  \label{prf:lem-thin}
  We have already argued for the base case, so let us turn to the
  inductive step where $\dim I_\ell<d$.  If $\ell>0$ and since our
  descending chain is strongly monotone, we can find an order ideal
  $I_{\ell-1}$ proper at step $\ell-1$ such that
  $\dim I_\ell\leq\dim I_{\ell-1}$.  Both if $\ell=0$
  or by repeating this argument, we obtain a sequence of
  order ideals (with decreasing indices)
  \begin{equation}\label{eq-greedy}
    I_\ell, I_{\ell-1},\dots,I_0
  \end{equation}
  where, for each $\ell>k\geq 0$, $I_k$ is proper at step $k$, and $\dim
  I_{k+1}\leq\dim I_k$.
  
  Let us decompose our sequence~\eqref{eq-greedy} by identifying the
  first step $L$ where $\dim I_{L+1}<\dim I_{L}$; let $L\eqdef -1$ if
  this never occurs.  After this step, for all $L\geq k\geq 0$, $\dim
  I_k>\dim I_\ell$.  Within the initial segment, for $\ell\geq k>L$,
  the dimension $\dim I_k$ remains constant equal to $\dim I_\ell$,
  and the induction hypothesis allows to apply \cref{cl-thin} and infer
  that every order ideal $I_k$ in this initial segment, and in
  particular~$I_\ell$ among them, is thin.
  
  It remains to provide a bound on~$\ell$.  The $\ell-L$
  order ideals in the initial segment are thin, and distinct since
  they are proper, hence by \cref{eq-card-thin},
  \begin{align}\label{eq-L}
    \ell &\leq L+\prod_{1\leq i\leq \fdim I_\ell}\hspace*{-.9em}
    (d-i+1)(N_{i}+1)\;.\ifams\hphantom{_{\fdim I_\ell} L_{\fdim I_\ell}}\hspace*{.65em}\hspace*{1.7em}\fi
  \end{align}
  \begin{description}
  \item[If {\boldmath$L\geq 0$}] we can apply the induction
      hypothesis to the proper order ideal $I_{L}$ of finite dimension $\fdim
      I_L<\fdim I_\ell$ along with \cref{NL-mono} to yield $L\leq L_{\fdim I_L}\leq L_{\fdim I_\ell-1}$ and therefore
    \begin{align}
    \ifams{}\else\hspace*{-1.5em}\fi\ell&\leq  L_{\fdim I_\ell-1} +\prod_{1\leq i\leq
      \fdim I_\ell}\hspace*{-.9em}(d-i+1)(N_{i}+1)\label{eq-L-ell-1}= L_{\fdim I_\ell}\;.
  \end{align}
  \item[If {\boldmath$L=-1$}] then \eqref{eq-L-ell-1} also
    holds since $L_{\fdim I_\ell-1}\geq 0>L$ in~\eqref{eq-L}.\qedhere\end{description}
\end{proof}

We deduce a general combinatorial statement on the length of
controlled strongly monotone descending chains, that generalises and
refines \cite[Thm.~4.4]{lazic21} thanks to thinness.

\begin{theorem}[Length function for strongly monotone descending chains]\label{th-thin}
  Consider a $(g,n_0)$-controlled strongly monotone descending
  chain~$D_0 \supsetneq \cdots\supsetneq D_\ell$ of downwards-closed
  subsets of~$\+N^d$. Then $\ell\leq L_{d}+1$.
\end{theorem}

\begin{proof}
In such a descending chain, either $\ell=0\leq L_{d}+1$, or
  $\ell>0$ and there must be an order ideal $I$ proper at step~$\ell-1$, and
  $I$ has finite dimension at most~$d$.  By
  \cref{lem-thin,NL-mono}, $\ell-1\leq L_{\fdim I}\leq
  L_{d}$ in that case.
\end{proof}

\subsection{Thin Order Ideals and Filters}\label{sub-thin-ideal}

Let us conclude this section with some consequences
of \cref{lem-thin} and \cref{cl-thin}.  Whereas thinness was
posited \emph{a priori} in the proof of
\citet[Thm.~3.3]{kunnemann23} and then shown to indeed allow a
suitable decomposition of minimal covering executions and to
eventually prove their result, here in the descending chain setting
it is an inherent property of all the order ideals appearing in the
chain, thereby providing a ``natural'' explanation for thinness.

\begin{corollary}\label{cor-thin}
  Consider a $(g,n_0)$-controlled strongly monotone descending
  chain $D_0 \supsetneq D_1\supsetneq\cdots$ of downwards-closed
  subsets of~$\+N^d$. Then every order ideal appearing in the chain is thin.
\end{corollary}

\cref{cor-thin} also entails a form of thinness of the minimal
  configurations in the complement of the downwards-closed sets
  $D_k$.  Recall that such a complement is the upward-closure of a finite basis
$B_k\eqdef\min_{\sqsubseteq}\+N^d\setminus D_k$.  Each element $\vec
v\in B_k$ is a vector defining a so-called \emph{(principal) order filter}
${\uparrow}\vec v$ of $\+N^d$.  Let us call a vector $\vec
v\in\+N^d$ \emph{nearly thin} if there exists a permutation
$\sigma\colon\{1,\dots,d\}\to\{1,\dots,d\}$ such that, for all $1\leq
i\leq d$, $\vec v(i)\leq N_{\sigma(i)}+1$.  We can relate thin order
ideals with nearly thin order filters, which by \cref{cor-thin}
applies to every vector $\vec v\in\bigcup_k B_k$\ifams\relax\else\
(see \ifx\authoranonymous\relax
\cref{app-thinfilter}\else the full version\fi\ for a proof)\fi.
\begin{restatable}{proposition}{thinfilter}\label{prop-thin-filter}
  If every order ideal in the canonical decomposition of a downwards-closed
  set $D\subseteq\+N^d$ is thin, then each $\vec
  v\in\min_{\sqsubseteq}\+N^d\setminus D$ is nearly thin.
\end{restatable}
\ifams\begin{proof}
\newcounter{myfoo}
\newcounter{mybar}
\setcounter{myfoo}{\value{equation}}
\setcounter{equation}{0}
\renewcommand{\theequation}{\fnsymbol{equation}}
Consider the canonical decomposition $D=I_1\cup\cdots\cup I_m$ of~$D$.
Then $U\eqdef\+N^d\setminus D=(\+N^d\setminus
I_1)\cap\cdots\cap(\+N^d\setminus I_m)$.  In turn, for each $1\leq
j\leq m$, $\+N^d\setminus
I_j=\bigcup_{i\in\fin(I_j)}{\uparrow}\big((I_j(i)+1)\cdot\vec e_i)$ where
$\vec e_i$ denotes the unit vector such that $\vec e_i(i)\eqdef 1$ and
$\vec e_i(j)\eqdef 0$ for all $j\neq i$.
Distributing intersections over unions, we obtain that
\begin{equation}\label{eq-U}
  U=\bigcup_{(i_1,\dots,i_m)\in\fin(I_1)\times\cdots\times\fin(I_m)}\bigcap_{1\leq j\leq m}{\uparrow}\big((I_j(i_j)+1)\cdot\vec e_{i_j})\;.
\end{equation}
For two order filters ${\uparrow}\vec v$ and ${\uparrow}\vec v'$,
$({\uparrow}\vec v)\cap({\uparrow}\vec v')={\uparrow}(\vec v\vee\vec
v')$ where $\vec v\vee\vec v'$ denotes the component-wise maximum of
$\vec v$ and $\vec v'$.  Therefore, by \eqref{eq-U}, any $\vec
v\in\min_\sqsubseteq U$ is of the form
\begin{equation}\label{eq-minv}
  \vec v_{i_1,\dots,i_m}\eqdef\bigvee_{1\leq j\leq m}\big((I_j(i_j)+1)\cdot\vec e_{i_j})
\end{equation}
for some $(i_1,\dots,i_m)\in\fin(I_1)\times\cdots\times\fin(I_m)$.
Note that not all the vectors $\vec v_{i_1,\dots,i_m}$ defined
by \eqref{eq-minv} are necessarily minimal in~$U$, but that
\begin{equation}\label{eq-U-V}
  \min_\sqsubseteq
U=\min_\sqsubseteq\{\vec v_{i_1,\dots,i_m}\mid
(i_1,\dots,i_m)\in\fin(I_1)\times\cdots\times\fin(I_m)\}\;.
\end{equation}
\medskip

Assume by contradiction that there exists some minimal vector $\vec
v\in\min_\sqsubseteq U$ that is not nearly thin.  Without loss of
generality, $\vec v(1)\leq\vec v(2)\leq\cdots\leq\vec v(d)$, as
otherwise we could apply a suitable permutation of $\{1,\dots,d\}$ on
the components of each ideal $I_j\in D$.  Then, because $N_i\leq
N_{i'}$ for all $i<i'$ by \cref{NL-mono}, $\vec v$ not being
nearly thin entails that there exists an index $k\in\{1,\dots,d\}$
such that $\vec v(k)>N_k+1$ but $\vec v(i)\leq N_{i}+1$ for all
$i<k$.

By \eqref{eq-U-V}, there exists
$(i_1,\dots,i_m)\in\fin(I_1)\times\cdots\times\fin(I_m)$ such that
$\vec v=\vec v_{i_1,\dots,i_m}$.  We are going to show that there
exists $(i'_1,\dots,i'_m)\in\fin(I_1)\times\cdots\times\fin(I_m)$ such
that $\vec v_{i'_1,\dots,i'_m}\sqsubsetneq\vec v$, which
by \eqref{eq-U-V} contradicts the minimality of~$\vec v$.

\medskip
Looking more closely at the individual components of~$\vec v$
in \eqref{eq-minv}, define for all $1\leq i\leq d$ the set
$S_i\eqdef\{1\leq j\leq m\mid i_j=i\}$ of indices $j\in\{1,\dots,m\}$
such that the value of $\vec v(i)$ ``stems'' from $I_j$.  Then
\addtocounter{equation}{3}
\begin{equation}\label{eq-minv-i}
  \vec v_{i_1,\dots,i_m}(i)=\begin{cases}0&\text{if }S_i=\emptyset\\
  \max\{I_{j}(i)+1\mid j\in S_i\}&\text{otherwise}.\end{cases}
\end{equation}
In particular, for the $k$th component, $S_k\neq\emptyset$ and we let
$V_k\eqdef\{j\in S_k\mid I_j(k)>N_k\}$ denote the indices~$j$ of the
ideals $I_j\in D$ responsible for the violation of near thinness.

\begin{example}[numbered=no,label=ex:thin]
  Let us illustrate the previous notations.  Let $d\eqdef 4$ and
  assume for the sake of simplicity that
  \begin{align*}
    N_1&\eqdef 2\;, & N_2&\eqdef4\;,&N_3&\eqdef 6\;,&N_4&\eqdef 8\;.
  \end{align*}
  Consider $D\eqdef\{I_1,I_2,I_3,I_4,I_5\}$ with
  \begin{align*}
    I_1&\eqdef(1,4,6,7)\;,&I_2&\eqdef(2,6,4,8)\;,&I_3&\eqdef(3,1,7,6)\;,\\
    I_4&\eqdef(3,1,7,6)\;,&I_5&\eqdef(4,5,3,0)\;.
  \end{align*}
  Then $\vec v\eqdef(2,7,7,7)=\vec v_{3,2,4,1,2}$ is not nearly thin
  with $k=2$, stem sets
  \begin{align*}
    S_1&=\{4\}\;, & S_2&=\{2,5\}\;, & S_3&=\{1\}\;,& S_4&=\{3\}\;,
  \end{align*}
  and $V_2=\{2,5\}$, and indeed $\vec v(2)$ stems from the ideals $I_2$
  and $I_5$, which are such that $I_2(2)=6>N_2$ and $I_5(2)=5>N_2$.\hfill\qedsymbol
\end{example}

For all $1\leq j\leq m$, because $I_j$ is thin, there exists a
bijection~$\sigma_j\colon\fin(I_j) \to\{1,\dots,\fdim I_j\}$ such
that, for all $i\in\fin(I_j)$, $I_j(i)\leq N_{\sigma_j(i)}$.  Without
loss of generality, we can assume that for all $i,i'\in\fin(I_j)$,
$I_j(i)\leq I_j(i')$ whenever $\sigma_j(i)<\sigma_j(i')$.
\begin{example}[numbered=no,continues=ex:thin,label=ex:thin2]
  Here are suitable bijections witnessing thinness:\pushQED{\qed}
  \begin{align*}
    \sigma_1&=(1\,2\,3\,4)\;,&
    \sigma_2&=(1\,3\,2\,4)\;,&
    \sigma_3&=(2\,1\,4\,3)\;,\\
    \sigma_4&=(2\,1\,4\,3)\;,&
    \sigma_5&=(3\,4\,2\,1)\;.&
            &\qedhere
  \end{align*}
\end{example}

For every $j\in V_k$,
$\sigma_j^{-1}(\{1,\dots,k\})\setminus\{1,\dots,k-1\}$ is non empty.
Therefore it contains an element $i'_j\geq k$ such that $I_j(i'_j)\leq
N_k$.  For every $1\leq j\leq m$ such that $j\not\in V_k$, let
$i'_j\eqdef i_j$.

Let us check that $\vec v_{i'_1,\dots,i'_m}\sqsubsetneq\vec v$, which
will allow to conclude.  Define $S'_i\eqdef\{1\leq j\leq m\mid
i'_j=i\}$ for each $1\leq i\leq d$; then \cref{eq-minv-i} holds
mutatis mutandis for $\vec v_{i'_1,\dots,i'_m}$ and
\begin{description}
\item[for $i<k$] $S'_i=S_i$ hence $\vec v_{i'_1,\dots,i'_m}(i)=\vec v(i)$;
\item[for $i=k$] $S'_k=S_k\setminus V_k=\{j\in S_k\mid I_j(k)\leq
  N_k\}$ hence $\vec v_{i'_1,\dots,i'_m}(k)\leq N_k+1<\vec v(k)$ by definition of~$k$;
\item[for $i>k$] $S'_i=S_i\cup\{j\in V_k\mid i'_j=i\}$ hence $\vec
  v_{i'_1,\dots,i'_m}(i)=\max\{I_j(i)+1\mid j\in
  S'_i\}=\max(\max\{I_j(i)+1\mid j\in S_i\},\max\{I_j(i)+1\mid
  j\in V_k\text{ and }i'_j=i\})$.
  \begin{itemize}
  \item On the one hand, 
    $\max\{I_j(i)+1\mid j\in S_i\}=\vec v(i)$.
  \item On the other hand, $I_j(i'_j)\leq N_k$ for
    all $j\in V_k$ by definition of $i'_j$, hence $\max\{I_j(i)+1\mid
    j\in V_k\text{ and }i'_j=i\}\leq N_k+1<\vec v(k)$ by definition of~$k$.
  \end{itemize}
  As $\vec v(k)\leq\vec v(i)$ by assumption since $i>k$, we
  conclude $\vec v_{i'_1,\dots,i'_m}(i)=\vec v(i)$.
\end{description}
\begin{example}[numbered=no,continues=ex:thin2]
  We have $\sigma_2^{-1}(\{1,2\})=\{1,3\}$ and
  $\sigma_5^{-1}(\{1,2\})=\{3,4\}$, hence we can pick $i'_2\eqdef 3$
  and $i'_5\eqdef 4$.  This defines $\vec
  v_{3,3,4,1,4}$ with stem sets
  \begin{align*}
    S'_1&=\{4\}\;,&
    S'_2&=\emptyset\;,&
    S'_3&=\{1,2\}\;,&
    S'_4&=\{3,5\}\;.
  \end{align*}
  Then $\vec v_{3,3,4,1,4}=(2,0,7,7)\sqsubsetneq\vec v$ as desired.\hfill\qedsymbol
\end{example}
\end{proof}

\setcounter{mybar}{\value{equation}}
\setcounter{equation}{\value{myfoo}}
\renewcommand{\theequation}{\arabic{equation}}
 \fi

\section{Applications}\label{sec-apps}

We describe two applications of \cref{th-thin} in this section.  The
first application in \cref{sec-vas} is to the coverability problem in
vector addition systems, and relies on the analysis of the backward
coverability algorithm done in~\citep{lazic21}.  Thus we can indeed
recover the improved upper bound of \citet{kunnemann23} for the
coverability problem in the more general setting of descending chains,
and show that the backward coverability algorithm
achieves this $n^{2^{O(d)}}$
upper bound (see \cref{cor-bc-vas}).

The second application in \cref{sec-an} focuses on the coverability
problem in invertible affine nets, a class introduced
by~\citet{benedikt17}, who analysed the complexity of the problem
through a reduction to zeroness in invertible polynomial automata.  We
give a direct analysis of the complexity of the backward coverability
algorithm, which follows the same lines as in the VAS case, and allows
to improve on the \ComplexityFont{2EXPSPACE} upper bound shown
in~\citep{benedikt17} for the problem, by showing that it is
actually \ComplexityFont{EXPSPACE}-complete (see \cref{cor-an}).  This
application additionally illustrates the usefulness of considering
strongly monotone descending chains rather than the $\omega$-monotone
ones, as the descending chains constructed by the backward algorithm
for invertible affine nets are in general not $\omega$-monotone.

As both applications take place in the framework of
well-structured transition systems~\citep{abdulla00,finkel01}, we
start with a quick refresher on this framework, the backward
coverability algorithm, and its dual view using downwards-closed
sets~\citep{lazic21} in the upcoming \cref{sec-wsts}.

\subsection{Coverability in Well-Structured Transition Systems}
\label{sec-wsts}
Well-structured transition systems (WSTS) form an abstract family of
computational models where the set of configurations is equipped with
a well-quasi-ordering ``compatible'' with the computation steps.  This
wqo ensures the termination of generic algorithms checking some
important behavioural properties like coverability and termination.
While the idea can be traced back to the 1980's~\citep{finkel87}, this
framework has been especially popularised through two landmark
surveys~\citep{abdulla00,finkel01} that emphasised its wide
applicability, and new WSTS models keep being invented in multiple
areas to this day.

\subsubsection{Well-Structured Transition Systems}
A \emph{well-structured transition system}
(WSTS)~\citep{abdulla00,finkel01} is a triple $(X,{\to},{\leq})$
where~$X$ is a set of configurations, ${\to}\subseteq X\times X$ is a
transition relation, and $(X,{\leq})$ is a wqo with the following
\emph{compatibility} condition: if $x\leq x'$ and $x\to y$, then there
exists $y'\geq y$ with $x'\to y'$.

The coverability problem below corresponds to the verification of
safety properties, i.e., to checking that no bad configuration can
ever be reached from a given initial configuration $s\in X$.  Here we
are given an error configuration $t\in X$, and we assume that any
configuration larger than~$t$ is also an error.
\begin{problem}[Coverability in well-structured transition systems]
\hfill
\begin{description}
\item[input] a well-structured transition system $(X,{\to},{\leq})$
and two configurations $s$ and $t$ in~$X$
\item[question] does $s$ \emph{cover} $t$, i.e., does
there exist $t'\in X$ such that $s\to^\ast t'\geq t$?
\end{description}
\end{problem}

\subsubsection{The Backward Coverability Algorithm}\label{sub-bc}
The first published version of this algorithm seems to date back
to~\citep{arnold78}, where it was used to show the decidability of
coverability in vector addition systems extended with reset
capabilities, before it was rediscovered and generalised to
well-structured transition systems~\citep{abdulla00}.

\paragraph{The Algorithm.}
Given an instance of the coverability problem, the \emph{backward
 coverability algorithm}~\citep{arnold78,abdulla00,finkel01} computes (a finite
 basis for) the upwards-closed set
\begin{equation}
 U_\ast\eqdef\{x\in X\mid\exists t'\geq t\mathbin. x\to^\ast t'\}
\end{equation}
of all the
configurations that cover~$t$, and then checks whether $s\in
U_\ast$.

The set~$U_\ast$ itself is computed by letting
\begin{align}\label{eq-bc}
  U_0&\eqdef{\uparrow}t\;,&
  U_{k+1}&\eqdef U_k\cup\mathrm{Pre}_\exists(U_k)\;,
\end{align}
  where, for a set $S\subseteq X$,
\ifams\begin{equation*}\else$\fi
  \mathrm{Pre}_\exists(S)\eqdef\{x\in X\mid\exists
                         y\in S\mathbin.x \to y\}.\ifams\end{equation*}\else$  \fi
Then $U_k=\{x\in X\mid\exists t'\geq
t\mathbin. x\to^{\leq k} t'\}$ is the set of configurations that can
cover~$t$ in at most~$k$ steps.
\Cref{eq-bc} defines a chain $U_0\subseteq U_1\subseteq\cdots$ of
upwards-closed subsets of~$X$.  Furthermore, if $U_\ell=U_{\ell+1}$ at
some step, then we have reached stabilisation:
$U_\ell=U_{\ell+k}=U_\ast$ for all~$k$.  Thus we focus in this
algorithm on ascending chains $U_0\subsetneq U_1\subsetneq\cdots$,
which are finite thanks to the ascending chain condition of the
wqo~$(X,{\leq})$.  In order to turn \eqref{eq-bc} into an actual
algorithm, one needs to make some effectiveness assumptions on
$(X,{\to},{\leq})$, typically that $\leq$ is decidable and a finite
basis for $\mathrm{Pre}_\exists({\uparrow}x)$ can be computed for all
$x\in X$~\citep[Prop.~3.5]{finkel01}.

\paragraph{A Dual View.}
\Citet{lazic21} take a dual view of the algorithm
and define from~\eqref{eq-bc} a descending chain $D_0\supsetneq
D_1\supsetneq\cdots$ of the same length where
\begin{equation}\label{eq-bc-D}
  D_k\eqdef X\setminus U_k
\end{equation}
for each~$k$; this stops with $D_\ast=X\setminus U_\ast$ the
set of configurations that do \emph{not} cover~$t$.  The entire computation
in~\eqref{eq-bc} can be recast in this dual view, by setting
\begin{align}\label{eq-bc-down}
  D_0&\eqdef X \setminus {\uparrow}t\;,&
  D_{k+1}&\eqdef D_k\cap\mathrm{Pre}_\forall(D_k)\;,
\end{align}
where, for a set $S\subseteq X$,
\ifams\begin{equation*}\else$\fi
  \mathrm{Pre}_\forall(S)\eqdef\{x\in X\mid\forall y\in
X\mathbin.(x\to y\implies y\in
S)\}=X\setminus(\mathrm{Pre}_\exists(X\setminus S)). \ifams\end{equation*}\else$  \fi
Under some effectiveness assumptions, in particular for manipulating
ideal representations over~$X$, this can be turned into an actual
algorithm~\citep[Sec.~3.1]{lazic21}.

\subsection{Coverability in Vector Addition Systems}\label{sec-vas}

Vector addition systems are a well-established model for simple
concurrent processes~\citep{karp69} equivalent to Petri nets, with
far-reaching connections to many topics in theoretical computer
science.  In particular, their coverability problem, which essentially
captures safety checking, has been thoroughly investigated from both a
theoretical~\citep{karp69,lipton76,rackoff78,bozzelli11,lazic21,kunnemann23}
and a more practical~\citep{esparza14,blondin16,geffroy18,blondin21}
standpoint.

\subsubsection{Vector Addition Systems}  
A $d$-dimensional \emph{vector addition system} (VAS)~\citep{karp69}
is a finite set $\vec A$ of vectors in $\+Z^d$.  It defines a
well-structured transition system $(\+N^d,{\to_{\vec
A}},{\sqsubseteq})$ with $\+N^d$ as set of configurations and
transitions $\vec u\to_{\vec A}\vec u+\vec a$ for all $\vec u$ in
$\+N^d$ and $\vec a$ in $\vec A$ such that $\vec u+\vec a$ is
in~$\+N^d$.  We work with a unary encoding, and let $\|\vec
u\|\eqdef\max_{1\leq i\leq d}|\vec u(i)|$ and $\|\vec
A\|\eqdef\max_{\vec a\in\vec A}\|\vec a\|$ for all $\vec u\in\+Z^d$
and $\vec A\subseteq\+Z^d$ finite.

The coverability problem in vector addition systems was first shown
decidable in \citeyear{karp69} by \citet{karp69}, before being
proven \EXPSPACE-complete when $d$ is part of the input
by \citet{lipton76} and \citet{rackoff78}.  Note that the problem
parameterised by~$d$ is trivial for $d=1$ (a target $\vec t$ is
coverable if and only if $\vec s\geq\vec t$ or there 
exists $\vec a\in\vec A$ with $\vec a>0$), hence we will assume $d\geq
2$.

\subsubsection{Complexity Upper Bounds}

The dual backward coverability algorithm of \cref{sub-bc} is
straightforward to instantiate in the case of a vector addition
system.  \Cref{fig-vas} displays the computed descending chain for the
2-dimensional VAS $\vec{A}_{\div 2} \eqdef \{(-2, 1)\}$ and target
configuration $\vec t\eqdef (0,5)$~\citep[Ex.~3.6]{lazic21}.

\begin{fact}[{\citep[claims~3.9 and~4.3]{lazic21}}]\label{fc-vas-ctrl}
  The descending chain $D_0\supsetneq D_1\supsetneq\cdots$
  defined by \crefrange{eq-bc}{eq-bc-down} for a $d$-dimensional
  VAS~$\vec A$ and a target vector~$\vec t$ is $(g,n_0)$-controlled
  for $g(x)\eqdef x+\|\vec A\|$ and $n_0\eqdef\|\vec t\|$, and is
  $\omega$-monotone.
\end{fact}

The length of the descending chain defined
by \crefrange{eq-bc}{eq-bc-down} is the main source of complexity for
the whole backward coverability algorithm, and we can apply our
own \cref{th-thin} instead of~\citep[Thm.~4.4]{lazic21} in order to
prove the following bound on this length, where the combinatorics are
somewhat similar to those of~\citep[Lem.~3.5]{kunnemann23}.

\begin{theorem}\label{th-vas}
  The backward coverability algorithm terminates after at most
  $n^{2^{O(d)}}$ iterations on a $d$-dimensional VAS encoded in unary.
\end{theorem}
\begin{proof}
  Let $n$ be the size of the input to the coverability problem; we
  assume in the following that $n,d\geq 2$.
  By \cref{fc-vas-ctrl} and due to the unary encoding, the descending
  chain $D_0\supsetneq D_1\supsetneq\cdots\supsetneq D_\ell=D_\ast$ is
  $(g,n_0)$-controlled for $g(x)\eqdef x+n$ and $n_0\eqdef n$, and is
  $\omega$-monotone and thus strongly monotone.  By \cref{th-thin},
  $\ell\leq L_{d}+1$.  Let us bound this value.

  \begin{restatable}{claim}{clthvas}\label{cl-thvas}
    Let $g(x)\eqdef x+n$ and $n_0\eqdef n$.  Then, for all $i\leq d$,
    \begin{align*}
    N_{i+1}&= n\cdot(L_i+2)\;,& L_i+4&\leq n^{3^i\cdot(\lg d+1)}\;.
    \end{align*}
  \end{restatable}  \begin{claimproof}[Proof of \Cref{cl-thvas}]
    \ifams\renewcommand{\qedsymbol}{\tiny[\ref{cl-thvas}]}\fi
  In the case of $N_{i+1}$, by
  the definition of $N_{i+1}$ in~\eqref{eqdef-Ni},
  $N_{i+1}=g^{L_{i}+1}(n_0)=n+(L_{i}+1)\cdot n=n\cdot(L_{i}+2)$ as
  desired.
  
  Regarding~$L_i$, we proceed by induction over~$i$.  For the base case
  $i=0$, $L_0+4=4\leq n^{3^0\cdot(\lg d+1)}$ since we assumed $n,d\geq 2$.  For
  the induction step, by the definition of $L_{i+1}$ in~\eqref{eqdef-Li}
  \begin{align*}
  L_{i+1}+4&=L_i+4+\prod_{0\leq j\leq i}(d-j)(N_{j+1}+1)\notag\\
  &\leq L_i+4+\prod_{0\leq j\leq i}(d-j)\cdot n\cdot(L_j+3)\notag\\
  &\leq 2\cdot(dn)^{i+1}\cdot\prod_{0\leq j\leq i}(L_j+3)\;.
  \intertext{Here, since $n\geq 2$,}
  2\cdot(dn)^{i+1} &\leq n^{(i+1)(\lg d+1)+1}
  \intertext{and by induction hypothesis for $j\leq i$}
  \prod_{0\leq j\leq i}(L_j+3)&\leq n^{\sum_{0\leq j\leq i}3^j(\lg d+1)}\;.
  \intertext{Thus, it only remains to see that, since $i>0$,}
  3^{i+1}\cdot(\lg d+1)&=(1+2\cdot\sum_{0\leq j\leq
    i}3^j)\cdot(\lg d+1)\notag\\
  &\geq (1+3^0+3^i)\cdot(\lg d+1)+\sum_{0\leq j\leq
    i}3^j\cdot(\lg d+1)\notag\\
  &\geq (i+1)\cdot(\lg d+1)+1+\sum_{0\leq j\leq
    i}3^j\cdot(\lg d+1)\;.\ifams\claimqedhere\else\\[-5.1em]\fi \end{align*}
  \end{claimproof}\ifams\relax\else\medskip\fi
   Thus $L_d+1\leq n^{3^d\cdot(\lg d+1)}$ by~\cref{cl-thvas}, which is
  in~$n^{2^{O(d)}}$.
\end{proof}

\begin{remark}[Branching or alternating vector addition systems]\label{rk-abvas}
The improved upper bound parameterised by the dimension~$d$
in \cref{th-vas} also applies to some extensions of vector addition
systems, for which \citet{lazic21} have shown that the backward
coverability algorithm was constructing an $\omega$-monotone
descending chain controlled as in \cref{fc-vas-ctrl}, namely
\begin{itemize}
\item in \citep[claims~6.7 and~6.8]{lazic21} for bottom-up
  coverability in branching vector addition systems~(BVAS)---which
  is \ComplexityFont{2EXP}-complete~\citep{demri12}---, and
\item in \citep[claims~5.4 and~5.5]{lazic21} for top-down coverability
  in alternating vector addition systems~(AVAS)---which
  is \ComplexityFont{2EXP}-complete as well~\citep{courtois14}.\hfill\qedsymbol
\end{itemize}
\end{remark}

Recall that $U_\ell$ is the set of configurations that can cover the
target~$\vec t$ in at most~$\ell$ steps, hence \cref{th-vas} provides
an alternative proof for \citep[Thm.~3.3]{kunnemann23}: if there exists a
covering execution, then there is one of length in $n^{2^{O(d)}}$, from which an
algorithm in $n^{2^{O(d)}}$ follows by~\citep[Thm.~3.2]{kunnemann23}.
Regarding the optimality of \cref{th-vas}, recall
that \citet{lipton76} shows an $n^{2^{\Omega(d)}}$ lower bound on the
length of a minimal covering execution, which translates into the same
lower bound on the number~$\ell$ of iterations of the backward
coverability algorithm~\citep[Cor.~2]{bozzelli11}.
Finally, this also yields an improved upper bound on the
complexity of the (original) backward coverability algorithm.  Here,
we can rely on the analysis performed by \citet[Sec.~3]{bozzelli11}
and simply replace \citeauthor{rackoff78}'s $n^{2^{O(d\lg d)}}$ bound on
the length of minimal covering executions by the bound
from \cref{th-vas}.
\begin{corollary}\label{cor-bc-vas}
  The backward coverability algorithm runs in time $n^{2^{O(d)}}$
  on $d$-dimensional VAS encoded in unary.
\end{corollary}
\begin{proof}
Let $n$ be the size of the input to the coverability problem and $U_0\subsetneq
U_1\subsetneq\cdots\subsetneq U_\ell=U_\ast$ be the ascending chain
constructed by the backward coverability according to \eqref{eq-bc}.
By \cref{th-vas}, $\ell$ is in $n^{2^{O(d)}}$.

Let $B_k\eqdef\min_\sqsubseteq U_k$ be the minimal basis at each
step~$k$.  The algorithm computes $B_{k+1}$ from $B_k$ as per
\eqref{eq-bc} by computing $\min_\sqsubseteq\mathrm{Pre}_\exists({\uparrow}\vec v)$
for each $\vec v\in B_k$, adding the elements of $B_k$, and removing
any non-minimal vector.  Thus each step can be performed in time
polynomial in~$n$, $d$, and the number of vectors in~$B_k$.  Here,
\citeauthor{bozzelli11}'s analysis in \citep[Sec.~3]{bozzelli11} shows
that $\|\vec v'\|\leq g(\|\vec v\|)$ for all $\vec
v'\in \min_\sqsubseteq\mathrm{Pre}_\exists({\uparrow}\vec v)$, 
yielding a bound of $|B_k|\leq (g^k(n)+1)^d\leq ((\ell + 1)\cdot n +
1)^d$, which is still in $n^{2^{O(d)}}$.

We can do slightly better.  By \cref{cor-thin}, all the ideals in
the canonical decomposition of $D_k\eqdef\+N^d\setminus U_k$ are thin,
and in turn \cref{prop-thin-filter} shows that all the vectors
in~$B_k$ are nearly thin.  Accordingly, let us denote by
$\mathrm{Fil}^{\mathsf{thin+1}}(\+N^d)$ the set of order filters
${\uparrow}\vec v$ such that $\vec v$ is nearly thin.  Then
$|B_k|\leq|\mathrm{Fil}^{\mathsf{thin+1}}(\+N^d)|$, and the latter is
in~$n^{2^{O(d)}}$:
\begin{align}
  |\mathrm{Fil}^{\mathsf{thin+1}}(\+N^d)|
  &\leq d!\cdot\prod_{1\leq i\leq d}(N_i+2)\notag\\
  &\leq d!\cdot n^d\cdot\prod_{0\leq i\leq d-1}(L_i+4)\tag{by \cref{cl-thvas} on $N_i$}\\
  &\leq n^{2d+\sum_{0\leq i\leq d-1}3^i\cdot(\lg d+1)}\tag{because $d\leq n$ and by \cref{cl-thvas} on $L_i$}\\
  &\leq n^{3^d\cdot(\lg d+1)}\;.\label{eq-fil}
\end{align}

Therefore, the overall complexity of the backward coverability
algorithm is polynomial in~$\ell$,
$\max_{0\leq k\leq \ell}|B_k|$, $n$, and~$d$, which is
in $n^{2^{O(d)}}$.
\ifams\par
Observe that the dual version of the backward coverability
algorithm enjoys the same upper bound: at each step~$k$, the algorithm
computes $D_{k+1}$ from $D_k$ as per \eqref{eq-bc-down}; this
computation of~$D_{k+1}$ can be performed in time polynomial in~$n$,
$d$, and the number of ideals in the canonical decomposition of
$D_k$~\citep[Sec.~3.2.1]{lazic21}.  By \cref{cor-thin,eq-card-thin},
$|D_k|\leq|\mathrm{Idl}^{\mathsf{thin}}(\+N^d)|\leq 1+L_d$, hence the
overall complexity of the dual algorithm is polynomial in~$L_d$, $n$,
and~$d$, which is still in $n^{2^{O(d)}}$.\fi
\end{proof}

The bounds in $n^{2^{O(d)}}$ for $\|\vec v\|\leq N_d+1$ for all $\vec
v\in\min_\sqsubseteq U_k$ and for $|\min_\sqsubseteq
U_k|\leq|\mathrm{Fil}^{\mathsf{thin+1}}(\+N^d)|$ in the previous proof
also improve on the corresponding bounds in~\citep[Thm.~9]{yen09}
and~\citep[Thm.~2]{bozzelli11}.  Recall
that \citet[Thm.~4.2]{kunnemann23} show that, assuming the exponential
time hypothesis, there does not exist a deterministic $n^{o(2^d)}$
time algorithm deciding coverability in unary encoded $d$-dimensional
VAS, hence the backward coverability algorithm
is conditionally optimal.
 
\subsection{Coverability in Affine Nets}\label{sec-an}
Affine nets~\cite{finkel04}, also known as affine vector addition
systems, are a broad generalisation of VAS and Petri nets encompassing
multiple extended VAS operations designed for greater modelling power.

\subsubsection{Affine Nets}
A $d$-dimensional (well-structured) \emph{affine net}~\cite{finkel04} is
a finite set $\?N$ of triples $(\vec a,A,\vec
b)\in\+N^d\times\+N^{d\times d}\times\+N^d$.  It defines a
well-structured transition system $(\+N^d,{\to_{\?N}},{\sqsubseteq})$
with $\+N^d$ as set of configurations and transitions $\vec u\to_{\?N}
A\cdot(\vec u-\vec a)+\vec b$ for all~$\vec u$ in~$\+N^d$ and $(\vec
a,A,\vec b)$ in~$\?N$ such that $\vec u-\vec a$ is in~$\+N^d$.  This
model encompasses notably
\begin{itemize}
  \item VAS and Petri nets when (each such)~$A$ is the identity matrix
  $I_d$,
  \item \emph{reset nets}~\cite{araki76,arnold78} when~$A$ is component-wise smaller
  or equal to~$I_d$,
  \item \emph{transfer nets}~\cite{ciardo94} when the sum of values in every
  column of~$A$ is one,
  \item \emph{post self-modifying nets}~\citep{valk78}---also known
  as \emph{strongly increasing affine
  nets}~\citep{finkel04,bonnet12}---when $A$ is component-wise larger
  or equal to~$I_d$, and
 \item \emph{invertible} affine nets~\citep{benedikt17} when $A$ is
  invertible over the rationals, i.e., $A\in\mathsf{GL}_d(\+Q)$.
\end{itemize}

As in the case of VAS, we will work with a unary encoding, and we let
$\|\?N\|\eqdef\max\{\|\vec a\|\mid (\vec a,A,\vec
b)\in\?N\}$; note that the entries from~$\vec b$ and~$A$ are not taken
into account.

\begin{example}\label{ex-an}
  Consider the affine nets
  \begin{align*}\hspace*{-1.2em}
    \?N_1&\eqdef\left\{\!{\footnotesize\left[\!\begin{array}{c}2\\0\end{array}\!\right],\left[\!\begin{array}{cc}1&0\\0&1\end{array}\!\right],\left[\!\begin{array}{c}0\\1\end{array}\!\right]}\!\right\}\\[.5em]
    \?N_2&\eqdef\left\{\!{\footnotesize\left[\!\begin{array}{c}0\\0\end{array}\!\right],\left[\!\begin{array}{cc}1&1\\0&0\end{array}\!\right],\left[\!\begin{array}{c}0\\0\end{array}\!\right]}\!\right\}\\[.5em]
    \?N_3&\eqdef\left\{\!{\footnotesize\left[\!\begin{array}{c}0\\0\end{array}\!\right],\left[\!\begin{array}{cc}1&1\\2&0\end{array}\!\right],\left[\!\begin{array}{c}0\\0\end{array}\!\right]}\!\right\}.  \end{align*}
Then $\?N_1$  defines the same WSTS as the 2-dimensional VAS $\vec A_{\div
  2}=\{(-2,1)\}$.  Focusing on the effects of their transition matrices,
  $\?N_2$
performs a transfer from its second component into its
  first component, while $\?N_3$ sums the values of its first two components into the first one, and
  puts the double of its first component into its second one.\hfill\qedsymbol
\end{example}

The coverability problem for reset VAS was first shown decidable
in \citeyear{arnold78} by \citet*{arnold78} using the backward
coverability algorithm, and the same algorithm applies to all affine
nets~\citep{dufourd98,finkel04}.  Its complexity is considerable:
their coverability problem has already an Ackermannian complexity in
the reset or transfer
cases~\citep{schnoebelen10,figueira11,icalp/Schmitz19}.  In the
strongly increasing case, \citet*[Lem.~11 and Thm.~13]{bonnet12} show
how to adapt \citeauthor{rackoff78}'s original argument to derive an
upper bound in~$n^{2^{O(d\lg d)}}$ on the length of minimal
coverability witnesses, with an \EXPSPACE\ upper bound for the problem
when~$d$ is part of the input, while in the invertible
case, \citet[Thm.~6]{benedikt17} show a \ComplexityFont{2EXPSPACE}
upper bound.

\paragraph{Control.}
Before we turn to the case of invertible affine nets, let us show that
the descending chains defined by the backward coverability algorithm
for affine nets are controlled, with a control very similar to the VAS
case (c.f.\ \cref{fc-vas-ctrl}).
\begin{proposition}\label{fc-an-ctrl}
  The descending chain $D_0\supsetneq D_1\supsetneq\cdots$ defined
  by \crefrange{eq-bc}{eq-bc-down} for a $d$-dimensional affine net~$\?N$ and a target
  vector~$\vec t$ is $(g,n_0)$-controlled for $g(x)\eqdef
  x+\|\?N\|$ and $n_0\eqdef\|\vec t\|$.
\end{proposition}
\begin{proof}
  Rather than handling $\mathrm{Pre}_\forall$ computations directly,
  we use the fact that
  $\mathrm{Pre}_\forall(S)=\+N^d\setminus(\mathrm{Pre}_\exists(\+N^d\setminus
  S))$ for all~$S\subseteq\+N^d$ and the following statement on
  $\mathrm{Pre}_\exists$ computations.
  \begin{claim}\label{cl-an-ctrl}
    If $\vec
  u'\in\min_{\sqsubseteq}\mathrm{Pre}_\exists({\uparrow}\vec u)$,
  then $\|\vec u'\|\leq\|\vec u\|+\|\?N\|$.
  \end{claim}
  \begin{claimproof}[Proof of \cref{cl-an-ctrl}]
    \ifams\renewcommand{\qedsymbol}{\tiny[\ref{cl-an-ctrl}]}\fi
    In such a situation, there exists a
    triple $(\vec a,A,\vec b)\in\?N$ such that $\vec u'\sqsupseteq\vec
  a$ and
    $A\cdot(\vec u'-\vec a)\sqsupseteq\vec u-\vec b$.  Let $\vec y$ be defined by
    $\vec y(i)\eqdef\max(\vec u(i),\vec b(i))-\vec b(i)$ for all
    $1\leq i\leq d$, thus of size $\|\vec y\|\leq\|\vec u\|$.  Then $\vec u'=\vec
    x+\vec a$ where $\vec x$ is a $\sqsubseteq$-minimal solution of the
    system of inequalities $A\vec x\sqsupseteq\vec y$.

    We are going to show that if $\vec x$ is an $\sqsubseteq$-minimal
    solution, then $\|\vec x\|\leq\|\vec y\|$.  This will yield the result, as then $\|\vec u'\|\leq \|\vec y\|+\|\vec a\|\leq \|\vec u\|+\|\?N\|$.
Assume by contradiction that $\vec x$ is a $\sqsubseteq$-minimal
    solution with $\vec x(j)>\|\vec y\|$ for some $1\leq j\leq d$.
    Consider $\vec x'$ defined by $\vec x'(j)\eqdef\|\vec y\|$ and
    $\vec x'(i)\eqdef\vec x(i)$ for all $i\neq j$; note that $\vec
    x'\sqsubsetneq\vec x$.  Let us show that $\vec x'$ is also a
    solution, i.e., that $A\vec x'\sqsupseteq\vec y$: for all $1\leq i\leq d$,
    \begin{itemize}
    \item if $A(i,j)>0$ then $\sum_{1\leq k\leq d}A(i,k)\cdot\vec
    x'(k)\geq\vec{x}'(j)\geq\|\vec y\|\geq\vec y(i)$, and
    \item otherwise $\sum_{1\leq k\leq d}A(i,k)\cdot\vec
    x'(k)=\sum_{1\leq k\leq d}A(i,k)\cdot\vec x(k)\geq\vec y(i)$ since
    $\vec x$ is a solution.
    \end{itemize}
    Thus $\vec x'$ is a solution, contradicting the $\sqsubseteq$-minimality
    of~$\vec x$.
  \end{claimproof}

  Now, since $D_0=\+N^d\setminus{\uparrow}\vec t$, $\|D_0\|\leq\|\vec
  t\|-1$ by \citep[Lem.~3.8]{lazic21}.  Regarding the control function,
  $D_{k+1}=D_k\cap\mathrm{Pre}_{\forall}(D_k)$ is such that
  $\|D_{k+1}\|\leq\max(\|D_k\|,\|\mathrm{Pre}_\forall(D_k)\|)$ also
  by \citep[Lem.~3.8]{lazic21}.  In turn, regarding
  $\mathrm{Pre}_\forall(D_k)=\+N^d\setminus\mathrm{Pre}_\exists(U_k)$,
  the minimal elements $\vec u$ of $U_k=\+N^d\setminus D_k$ have size
  $\|\vec u\|\leq\|D_k\|+1$ still by \citep[Lem.~3.8]{lazic21}, thus
  the minimal elements $\vec u'$ of $\mathrm{Pre}_\exists(U_k)$ have
  size $\|\vec u'\|\leq\|D_k\|+1+\|\?N\|$ by \cref{cl-an-ctrl}, hence
  $\|\mathrm{Pre}_\forall(D_k)\|\leq \|D_k\|+\|\?N\|$ by a last
  application of~\citep[Lem.~3.8]{lazic21}.
\end{proof}

\subsubsection{Invertible Affine Nets}\label{sub-inv}

The restriction to invertible
affine nets~\citep{benedikt17} is somehow
orthogonal to the usual restrictions to reset/transfer/post self-modifying/\dots\ nets.
For instance, in \cref{ex-an}, the identity matrix in~$\?N_1$ is
clearly invertible, and the transfer matrix in~$\?N_2$ is not.  More
generally, reset nets are never invertible (when they perform resets),
and transfer nets are invertible exactly when their matrices are
permutation matrices.  Nevertheless, some more involved affine nets
are invertible, like $\?N_3$ in \cref{ex-an}, whose matrix is
invertible with inverse
${\footnotesize\left[\!\begin{array}{cc}0&1/2\\1&-1/2\end{array}\!\right]}$.

\paragraph{Strong Monotonicity.}
When dealing with a descending sequence of downwards-closed sets
produced by the dual backward coverability algorithm for WSTS, a key
observation made in~\citep{lazic21} allows to sometimes derive
monotonocity.
For this, in a WSTS $(X,{\to},{\leq})$, define
$\mathrm{Post}_\exists(S)\eqdef\{y\in X\mid\exists x\in S\mathbin.x\to
y\}$.  Following~\cite{blondin18}, for two order ideals $I$ and
$I'$, write $I\idealto I'$ if $I'$ appears in the canonical decomposition
of ${\downarrow}\mathrm{Post}_\exists(I)$.

\begin{fact}[{\citep[Claim~4.2]{lazic21}}]\label{fc-proper-tr}
  Let $D_0\supsetneq D_1\supsetneq\cdots$ be a descending chain of
  downwards-closed sets defined by~\crefrange{eq-bc}{eq-bc-down}.  If
  $I_{k+1}$ is a ideal proper at step~$k+1$, then there exists an order
  ideal~$I$ and an order ideal~$I_k$ proper at step~$k$ such that
  $I_{k+1}\idealto I\subseteq I_k$.
\end{fact}

In the case of affine
nets, and identifying order ideals~$I$ with vectors in~$\+N^d_\omega$
with $\omega+n=\omega-n=\omega\cdot n=\omega$ for all $n$ in~$\+N$,
${\downarrow}\mathrm{Post}_\exists(I)={\downarrow}\{A\cdot(I-\vec
a)+\vec b\mid(\vec a,A,\vec b)\in\?N,I\sqsupseteq\vec a\}$.

\begin{proposition}\label{fc-an-mono}
  The descending chain $D_0\supsetneq D_1\supsetneq\cdots$ defined by
  \crefrange{eq-bc}{eq-bc-down} for a $d$-dimensional invertible affine net~$\?N$ and a target
  vector~$\vec t$ is strongly monotone.
\end{proposition}
\begin{proof}
  Let $I_{k+1}$ be proper at step~$k+1$.  By \cref{fc-proper-tr},
  there exists an order ideal~$I$ and an order ideal~$I_k$ proper at
  step~$k$ such that $I_{k+1}\idealto_\?N I\subseteq I_k$.  Let us show
  that $\dim I_{k+1}\leq\dim I$; as $\dim I\leq\dim I_k$ because
  $I\subseteq I_k$, this will yield the result.

  Since $I_{k+1}\idealto_\?N I$, there exists $(\vec a,A,\vec b)$ in~$\?N$
  such that $I-\vec b=A\cdot(I_{k+1}-\vec a)$.  For this to hold, note
  that for all $i\in\fin(I)$, the $i$th row of~$A$ must be such that
  $A(i,j)=0$ for all $j\in\omega(I_{k+1})$.  As $A$ is invertible,
  those $(\fdim I)$-many rows must be linearly independent.  As just
  argued, the $j$th column for each of these rows is made of zeroes
  whenever $j\in\omega(I_{k+1})$.  Thus the remaining $(\fdim
  I_{k+1})$-many columns must make those $\fdim I$ rows linearly
  independent, hence necessarily $\fdim I_{k+1}\geq\fdim I$, i.e.,
  $\dim I_{k+1}\leq\dim I$.
\end{proof}
Observe that the proof of \cref{fc-an-mono} does not work for the
transfer net~$\?N_2$ of \cref{ex-an}:
${\footnotesize\left[\!\begin{array}{c}\omega\\\omega\end{array}\!\right]}\idealto_{\?N_2}{\footnotesize\left[\!\begin{array}{c}\omega\\0\end{array}\!\right]}$;
this is exactly the kind of non-monotone behaviour invertibility was
designed to prevent.  Also observe that
${\footnotesize\left[\!\begin{array}{c}2\\\omega\end{array}\!\right]}\idealto_{\?N_3}{\footnotesize\left[\!\begin{array}{c}\omega\\4\end{array}\!\right]}$
in the invertible affine net~$\?N_3$, which is not an
$\omega$-monotone behaviour: this illustrates the usefulness of
capturing strongly monotone descending chains, as \citep[Thm.~4.4 and
Cor.~4.6]{lazic21} do not apply.

\paragraph{Complexity Upper Bounds.} We are now equipped to analyse
the complexity of the backward coverability algorithm in invertible
affine nets.  Regarding the length~$\ell$ of the chain constructed by
the algorithm, by \cref{fc-an-ctrl,fc-an-mono} we are in the same
situation as in \cref{th-vas} and we can simply repeat the arguments
from its proof.

\begin{theorem}\label{th-an}
  The backward coverability algorithm terminates after at most
  $n^{2^{O(d)}}$ iterations on $d$-dimensional invertible affine nets
  encoded in unary when $d\geq 2$.
\end{theorem}

We deduce two corollaries from \cref{th-an}: one pertaining to the
complexity of the backward coverability algorithm in
dimension~$d$, which mirrors \cref{cor-bc-vas}, and one for the
coverability problem when $d$ is part of the input.  Let us start with the
backward coverability algorithm.

\begin{restatable}{corollary}{coran}\label{cor-an-bc}
  The backward coverability algorithm runs in time
  $n^{2^{O(d)}}$ on $d$-dimensional invertible affine nets encoded
  in unary when $d\geq 2$.
\end{restatable}
\begin{proof}
  \Cref{th-an} shows that the length~$\ell$ of the ascending chain
  $U_0\subsetneq U_1\subsetneq\cdots\subsetneq U_\ell=U_\ast$ constructed by
  the backward coverability algorithm is at most $L_d+1$, which
  is in $n^{2^{O(d)}}$.

  Let $B_k\eqdef\min_\sqsubseteq U_k$ denote the minimal basis at
  step~$k$. In order to compute $B_{k+1}$ as per \eqref{eq-bc}, thanks
  to \cref{cl-an-ctrl}, we could essentially argue as in the proof
  of \cref{cor-bc-vas}, with the caveat that computing bluntly
  $\min_\sqsubseteq\mathrm{Pre}_\exists({\uparrow}\vec v)$ for each
  $\vec v\in B_k$ is dangerously similar to a linear integer
  programming question and will incur an additional cost.

  Alternatively, recall from \cref{eq-fil} that
  $\mathrm{Fil}^{\mathsf{thin+1}}(\+N^d)$, the set of order filters
  ${\uparrow}\vec v$ such that $\vec v$ is nearly thin, has
  at most $n^{2^{O(d)}}$ elements, and that
  $|B_k|\leq|\mathrm{Fil}^{\mathsf{thin+1}}(\+N^d)|$
  by \cref{cor-thin,prop-thin-filter}.  Thus in order to compute
  $B_{k+1}$ one can enumerate the nearly thin vectors $\vec
  v'\in\mathrm{Fil}^{\mathsf{thin+1}}(\+N^d)$ and check for each
  $(\vec a,A,\vec b)\in\?N$ such that $\vec v'\sqsupseteq\vec a$
  whether there exists $\vec v\in B_k$ such that $A\cdot(\vec v'-\vec
  a)+\vec b\sqsupseteq\vec v$.  Each such check can be performed in
  time polynomial in~$\|\vec v'\|\leq N_{d}+1=n\cdot(L_{d-1}+2)+1$,
  $n$, $d$, and $|B_k|\leq|\mathrm{Fil}^\mathsf{thin+1}(\+N^d)|$.
  Thus the entire computation can be carried out in $n^{2^{O(d)}}$.
\ifams\par The same upper bound holds for the dual version of the backward
  coverability algorithm.  At each step~$k$, by \cref{cor-thin}, in
  order to compute $D_{k+1}$ one can enumerate the thin order ideals
  $I\in\mathrm{Idl}^\mathsf{thin}(\+N^d)$ and check for each such~$I$
  whether $I\subseteq D_k$ and $I\subseteq\mathrm{Pre}_\forall(D_k)$,
  before removing the non-maximal ones.  Note that
  $I\subseteq\mathrm{Pre}_\forall(D_k)$ if and only if
  $\mathrm{Post}_\exists(I)\subseteq D_k$, if and only if
  $\{A\cdot(I-\vec a)+\vec b\mid(\vec a,A,\vec
  b)\in\?N,I\sqsupseteq\vec a\}\subseteq D_k$, which can be checked in
  time polynomial in~$\|I\|\leq N_{d}=n\cdot(L_{d-1}+2)$, $n$, $d$,
  and $|D_k|\leq|\mathrm{Idl}^\mathsf{thin}(\+N^d)|\leq L_d+1$
  by \cref{eq-card-thin}.  The entire computation can be performed in
  time polynomial in~$L_d$, $n$, and $d$, and this remains in
  $n^{2^{O(d)}}$.\fi
\end{proof}
As VAS are a particular case of invertible affine nets, the upper
bounds in \cref{cor-an-bc} are optimal assuming the
exponential time hypothesis by~\citep[Thm.~4.2]{kunnemann23}.

\medskip
Our last result concerns the complexity of coverability in invertible
affine nets when~$d$ is part of the input.  Note that the arguments
leading to an algorithm working in space $O(d\lg(n\cdot\ell))$ in the
VAS case~\citep[Thm.~3.2]{kunnemann23}---which are essentially the
same as those used to derive a \ComplexityFont{2EXPSPACE} upper bound
for invertible affine nets in~\citep[Thm.~6]{benedikt17}---do not work
here, as the configurations along an execution of an affine net can
grow exponentially with~$\ell$.

\begin{restatable}{corollary}{than}\label{cor-an}
  The coverability problem for invertible affine nets is
  \ComplexityFont{EXPSPACE}-complete.
\end{restatable}
\begin{proof}
  The hardness for \ComplexityFont{EXPSPACE} follows from the hardness
  of the coverability problem for VAS~\citep{lipton76}.  \ifams\par\fi
  Regarding the upper bound, consider the execution of the classical
  backward coverability algorithm as defined in \cref{eq-bc} on an
  invertible affine net~$\?N$ with target configuration~$\vec t$: this
  is an ascending chain $U_0\subsetneq U_1\subsetneq\cdots\subsetneq
  U_\ell$ where $U_\ell=U_{\ell+1}=U_\ast$.  The following
  characterisation of coverability actually holds more generally in
  WSTS.

  \begin{claim}\label{cl-pseudo-wit} In an affine net~$\?N$, $\vec s$
  covers $\vec t$ if and only if there exists $\ell'\leq\ell$ and a sequence of 
  configurations $\vec t_0,\dots,\vec t_{\ell'}$, called
  a \emph{coverability pseudo-witness}, satisfying
  \begin{align}
    \vec t_0&\eqdef\vec t\;,&
    \vec
  t_{k+1}&\in\min_{\sqsubseteq}\mathrm{Pre}_\exists({\uparrow}\vec
  t_k)\;,&
  \vec t_{\ell'}&\sqsubseteq\vec s\;.\label{eq-pseudo-wit}
  \end{align}
  \end{claim}
  \begin{claimproof}[Proof of \Cref{cl-pseudo-wit}]
    \ifams\renewcommand{\qedsymbol}{\tiny[\ref{cl-pseudo-wit}]}\fi
  If a coverability pseudo-witness exists, then we claim that for all
  $\ell'\geq k\geq 0$ there exists $\vec s_k\sqsupseteq\vec t_k$ such
  that $\vec s=\vec s_{\ell'}\to_\?N\vec
  s_{\ell'-1}\to_\?N\cdots\to_\?N\vec s_k$, and thus in particular $\vec
  s\to_\?N^\ast\vec s_{0}\geq\vec t_0$ for $k=0$.  We can check this by
  induction over~$k$.  For the base case $k=\ell'$, define $\vec
  s_{\ell'}\eqdef\vec s$.  For the induction step $k$, since
  $\vec t_{k+1}\in\mathrm{Pre}_\exists({\uparrow}\vec t_k)$ there exists
  $\vec s'_k\sqsupseteq\vec t_k$ such that $\vec t_{k+1}\to_\?N\vec
  s'_k$; by WSTS compatibility and since $\vec s_{k+1}\sqsupseteq\vec
  t_{k+1}$, there exists $\vec s_k\sqsupseteq\vec s'_k$ such that
  $\vec s_{k+1}\to_\?N\vec s_k$.

  Conversely, assume that $\vec s$ covers $\vec t$ in~$\?N$. Then
  $\vec s\in U_\ell$, and let $\ell'\leq\ell$ be the least index such
  that $\vec s\in U_{\ell'}$.  Then either $\ell'=0$, i.e., $\vec
  s\sqsupseteq\vec t=\vec t_0$ and we are done, or $\ell'>0$. Because
  $\vec s\in U_{\ell'}$ there must be
  some $\vec t_{\ell'}\in\min_{\sqsubseteq} U_{\ell'}$ with $\vec
  s\sqsupseteq\vec t_{\ell'}$, and $\vec t_{\ell'}\not\in U_{\ell'-1}$
  as otherwise $\vec s$ would be in $U_{\ell'-1}$, contradicting the
  minimality of~$\ell'$.  In general, if we have
  found a sequence $(\vec t_j)_{\ell'\geq j\geq k>0}$
  satisfying~\eqref{eq-pseudo-wit} until rank~$k+1$ included and know
  that $\vec t_{k}\in(\min_{\sqsubseteq} U_{k})\setminus U_{k-1}$,
  then either $k=1$ and $\vec
  t_1\in\min_{\sqsubseteq}\mathrm{Pre}_\exists({\uparrow}\vec t_0)$ by
  definition of $U_0$ and~$U_1$ in \eqref{eq-bc}, or $k>1$ and because
  $\vec t_k\not\in U_{k-1}$, there exists $\vec
  t_{k-1}\in\min_{\sqsubseteq} U_{k-1}$ such that
  $\vec t_k\in\min_{\sqsubseteq}\mathrm{Pre}_\exists({\uparrow}\vec
  t_{k-1})$, and $\vec t_{k-1}\not\in U_{k-2}$ as otherwise we would
  have $\vec t_k$ in~$U_{k-1}$.  Repeating this process yields a coverability
  pseudo-witness.
  \end{claimproof}
  
  By \cref{cl-pseudo-wit}, a non-deterministic algorithm for
  coverability can guess and check the existence of a coverability
  pseudo-witness.  By \cref{th-an}, such a pseudo-witness has a length
  $\ell'\leq\ell$ in $n^{2^{O(d)}}$.  Furthermore,
  by \cref{cl-an-ctrl} the components in each $\vec t_k$ in such a
  pseudo-witness are bounded by $\|\vec t\|+\|\?N\|\cdot k\leq
  (\ell+1)\cdot n$, which is still in $n^{2^{O(d)}}$.  Thus
  exponential space suffices.  Note that this also holds when we
  assume the invertible affine net to be encoded in binary, by
  substituting $2^n$ for $n$ in the bound $n^{2^{O(d)}}$.
\end{proof}

\begin{remark}[Strictly increasing affine nets]\label{rk-sian}
  Strictly increasing affine nets~\citep{valk78,finkel04,bonnet12} are
  intuitively the affine nets devoid of any form of reset or transfer;
  in \cref{ex-an}, only~$\?N_1$ is strictly increasing.  All the
  results we have proven for invertible affine nets in this
  section---namely in \cref{th-an,cor-an-bc,cor-an}---also hold for
  strictly increasing affine nets, because the descending chains of
  downwards-closed sets they generate when running the backward
  coverability algorithm are
  $\omega$-monotone. 

\begin{claim}\label{fc-sian-mono}
  The descending chain $D_0\supsetneq D_1\supsetneq\cdots$ defined by
  \crefrange{eq-bc}{eq-bc-down} for a $d$-dimensional strictly increasing affine net~$\?N$ and a target
  vector~$\vec t$ is $\omega$-monotone.
\end{claim}
\begin{claimproof}[Proof of \cref{fc-sian-mono}]
    \ifams\renewcommand{\qedsymbol}{\tiny[\ref{fc-sian-mono}]}\fi
  Let $I_{k+1}$ be proper at step~$k+1$.  By \cref{fc-proper-tr},
  there exists an order ideal~$I$ and an order ideal~$I_k$ proper at
  step~$k$ such that $I_{k+1}\idealto_\?N I\subseteq I_k$.  Let us show
  that $\omega(I_{k+1})\subseteq\omega(I)$; as
  $\omega(I)\subseteq\omega(I_k)$ because $I\subseteq I_k$, this will
  yield the result.

  Since $I_{k+1}\idealto_\?N I$, there exists $(\vec a,A,\vec b)$ in~$\?N$
  such that $I_{k+1}\sqsupseteq\vec a$ and $I=A\cdot(I_{k+1}-\vec a)+\vec b$.
  Because~$\?N$ is strictly increasing, $A=I_d+A'$ for some matrix
  $A'\in\+N^{d\times d}$, hence $I=I_{k+1}-\vec a+A'\cdot(I_{k+1}-\vec
  a)+\vec b$.  Thus $I\sqsupseteq(I_{k+1}-\vec a)$ and therefore
  $\omega(I)\supseteq\omega(I_{k+1})$.
\end{claimproof}

 An \EXPSPACE\ upper bound was
  already shown by \citet{bonnet12} for the coverability problem, but the $n^{2^{O(d)}}$ bound
  for the problem parameterised by~$d$ is an
  improvement over the $n^{2^{O(d\lg d)}}$ bounds of \citep[Lem.~11 and
  Thm.~13]{bonnet12}, and the bounds for the backward coverability algorithm are new.\hfill\qedsymbol
\end{remark}

\bibliographystyle{abbrvnat}
\bibliography{conferences,journalsabbr,refs}

\end{document}